\documentclass[a4paper]{article}

\usepackage{natbib}
\usepackage{amsbsy,amsfonts,amsmath, amssymb, amsthm, mathtools}
\usepackage{tikz}
\usepackage{float}
\usepackage{color}

\floatstyle{ruled}
\newfloat{algorithm}{htb}{alg}
\floatname{algorithm}{Algorithm}

\newcommand{\red}[1]{#1}

\newcommand{\R}{\mathbb{R}}
\newcommand{\Z}{\mathbb{Z}}

\newtheorem{thm}{Theorem}   
\newtheorem{lem}{Lemma}    
\newtheorem{cnj}{Conjecture} 
\newtheorem{rem}{Remark}   

\theoremstyle{definition}
\newtheorem{example}{Example}   

\usetikzlibrary{arrows,decorations.markings,decorations.pathreplacing,patterns,calc, shapes}

\def\keywords{\vspace{-.3em}
 \if@twocolumn
  \small\it Keywords\/\bf---$\!$%
 \else
\begin{center}\small\bf
Keywords\end{center}\quotation\small  \fi}
\def\endkeywords{\vspace{0.6em}\par\if@twocolumn\else\endquotation\fi
 \normalsize\rm}

\title{A reclaimer scheduling problem arising in coal stockyard management\thanks{This research was supported by ARC Linkage Grant nos. LP0990739 and LP110200524, and by the HVCCC and Triple Point Technology.}}

\author{
Enrico Angelelli \\
{\it University of Brescia, Italy} \and
Thomas Kalinowski
, Reena Kapoor\\
{\it University of Newcastle, Australia} \and
Martin W.P. Savelsbergh\\
{\it Georgia Institute of Technology, USA}
}

\begin{document}

\maketitle

\begin{abstract}

We study a number of variants of an abstract scheduling problem
inspired by the scheduling of reclaimers in the stockyard of a coal
export terminal. We analyze the complexity of each of the variants,
providing complexity proofs for some and polynomial algorithms for
others. For one, especially interesting variant, we also develop a
constant factor approximation algorithm.

\medskip
\noindent\textbf{Keywords:} reclaimer scheduling, stockyard management, approximation algorithm, complexity
\end{abstract}


\section{Introduction}

We investigate a scheduling problem that arises in the management of
a stockyard at a coal export terminal. Coal is marketed and sold to
customers by brand. The brand of coal dictates its characteristics,
for example the range in which its calorific value, ash, moisture
and/or sulphur content lies. In order to deliver a brand required by
a customer, coal from different mines, producing coal with different
characteristics, is ``mixed'' in a stockpile at the stockyard of a
coal terminal to obtain a blended product meeting the required brand
characteristics. Stackers are used to add coal that arrives at the
terminal to stockpiles in the yard and reclaimers are used to
reclaim completed stockpiles for delivery to waiting ships at the
berths.

We focus on the scheduling of the reclaimers. The stockyard
motivating our investigation has four pads on which stockpiles are
build. A stockpile takes up the entire width of a pad and a portion
of its length. Each pad is served by two reclaimers that cannot pass
each other and each reclaimer serves two pads, one on either side of
the reclaimer. Effective reclaimer scheduling, even though only one
of component of the management of the stockyard management at a coal
terminal, is a critical component, because reclaimers tend to be the
constraining entities in a coal terminal (reclaiming capacity, in
terms of tonnes per hour, is substantially lower than stacking
capacity).

In order to gain a better understanding of the challenges associated
with reclaimer scheduling, we introduce an abstract model of
reclaimer scheduling and study the complexity of different variants
of the model as well as algorithms for the solution of these
variants. Our investigation has not only resulted in insights that
may be helpful in improving stockyard efficiency, but has also given
rise to a new and intriguing class of scheduling problems that has
proven to be surprisingly rich. One reason is that the travel time
of the reclaimers, i.e., the time between the completion of the
reclaiming of one stockpile and the start of the reclaiming of a
subsequent stockpile, cannot be ignored. Another reason is the
interaction between the two reclaimers, caused by the fact that they
cannot pass each other.

The remainder of the paper is organized as follows.  In Section
\ref{sec:background}, we provide background information on the
operation of a coal export terminal and the origin of the reclaimer
scheduling problem. In Section \ref{sec:lit}, we provide a brief
literature review. In Section \ref{sec:problem}, we introduce the
abstract model of the reclaimer scheduling problem that is the focus
of our research and we introduce a graphical representation of
schedules that will be used throughout the paper. In Sections
\ref{sec:without} and \ref{sec:with}, we present the analysis of a
number of variants of the reclaimer scheduling problem. In Section
\ref{sec:final}, we give some final remarks and discuss future
research opportunities.

\section{Background}
\label{sec:background}

The Hunter Valley Coal Chain (HVCC) refers to the inland portion of
the coal export supply chain in the Hunter Valley, New South Wales,
Australia, which is the largest coal export supply chain in the
world in terms of volume. Most of the coal mines in the Hunter
Valley are open pit mines. The coal is mined and stored either at a
railway siding located at the mine or at a coal loading facility
used by several mines. The coal is then transported to one of the
terminals at the Port of Newcastle, almost exclusively by rail. The
coal is dumped and stacked at a terminal to form stockpiles. Coal
from different mines with different characteristics is ``mixed'' in
a stockpile to form a coal blend that meets the specifications of a
customer. Once a vessel arrives at a berth at the terminal, the
stockpiles with coal for the vessel are reclaimed and loaded onto
the vessel. The vessel then transports the coal to its destination.
The coordination of the logistics in the Hunter Valley is
challenging as it is a complex system involving 14 producers
operating 35 coal mines, 27 coal load points, 2 rail track owners, 4
above rail operators, 3 coal loading terminals with a total of 8
berths, and 9 vessel operators. Approximately 1700 vessels are
loaded at the terminals in the Port of Newcastle each year. For a
more in-depth description of the Hunter Valley Coal Chain see Boland
and Savelsbergh (\cite{Boland}).

An important characteristic of a coal loading terminal is whether it
operates as a cargo assembly terminal or as a dedicated stockpiling
terminal. When a terminal operates as a cargo assembly terminal, it
operates in a ``pull-based'' manner, where the coal blends assembled
and stockpiled are based on the demands of the arriving ships.  When
a terminal operates as dedicated stockpiling terminal, it operates
in a ``push-based'' manner, where a small number of coal blends are
built in dedicated stockpiles and only these coal blends can be
requested by arriving vessels.  We focus on cargo assembly terminals
as they are more difficult to manage due to the large variety of
coal blends that needs to be accommodated.

Depending on the size and the blend of a cargo, the assembly may
take anywhere from three to seven days.  This is due, in part, to
the fact that mines can be located hundreds of miles away from the
port and getting a trainload of coal to the port takes a
considerable amount of time. Once the assembly of a stockpile has
started, it is rare that the location of the stockpile in the
stockyard is changed; relocating a stockpile is time-consuming and
requires resources that can be used to assemble or reclaim other
stockpiles. Thus, deciding where to locate a stockpile and when to
start its assembly is critical for the efficiency of the system.
Ideally, the assembly of the stockpiles for a vessel completes at
the time the vessel arrives at a berth (i.e., ``just-in-time''
assembly) and the reclaiming of the stockpiles commences
immediately. Unfortunately, this does not always happen due to the
limited capacities of the resources in the system, e.g., stockyard
space, stackers, and reclaimers, and the complexity of the stockyard
planning problem.

A seemingly small, but in fact crucial component of the planning
process is the scheduling of the reclaimers, because reclaimers tend
to be the constraining entities in a coal terminal (the reclaiming
capacity is substantially lower than the stacking capacity).

The characteristics of the reclaimer scheduling problems studied in
this paper are motivated by those encountered at a stockyard at one
of the cargo assembly terminals at the Port of Newcastle. At this
particular terminal, the stockyard has four pads, $A$, $B$, $C$, and
$D$, on which cargoes are assembled. Coal arrives at the terminal by
train. Upon arrival at the terminal, a train dumps its contents at
one of three dump stations. The coal is then transported on a
conveyor to one of the pads where it is added to a stockpile by a
stacker. There are six stackers, two that serve pad $A$, two that
serve pad $B$ and pad $C$, and two that serve pad $D$. A single
stockpile is built from several train loads over several days. After
a stockpile is completely built, it dwells on its pad for some time
(perhaps several days) until the vessel onto which it is to be
loaded is available at one of the berths. A stockpile is reclaimed
using a bucket-wheel reclaimer and the coal transferred to the berth
on a conveyor. The coal is then loaded onto the vessel by a
shiploader. There are four reclaimers, two that serve pad $A$ and
pad $B$ and two that serve pad $C$ and pad $D$. Both stackers and
reclaimers travel on rails at the side of a pad. Stackers and
reclaimers that serve that same pads cannot pass each other.

A brief overview of the events driving the cargo assembly planning
process is presented next.  An incoming vessel alerts the coal chain
managers of its pending arrival at the port. This announcement is
referred to as the vessel's nomination.  Upon nomination, a vessel
provides its estimated time of arrival ($ETA$) and a specification
of the cargoes to be assembled to the coal chain managers.  As coal
is a blended product, the specification includes for each cargo a
recipe indicating from which mines coal needs to be sourced and in
what quantities.  At this time, the assembly of the cargoes
(stockpiles) for the vessel can commence.  A vessel cannot arrive at
a berth prior to its $ETA$, and often a vessel has to wait until
after its $ETA$ for a berth to become available.  Once at a berth,
and once all its cargoes have been assembled, the reclaiming of the
stockpiles (the loading of the vessel) can begin.  A vessel must be
loaded in a way that maintains its physical balance in the water. As
a consequence, for vessels with multiple cargoes, there is a
predetermined sequence in which its cargoes must be reclaimed. The
goal of the planning process is to maximize the throughput without
causing unacceptable delays for the vessels.

For a given set of vessels arriving at the terminal, the goal is
thus to assign each cargo of a vessel to a location in the
stockyard, schedule the assembly of these cargoes, and schedule the
reclaiming of these cargoes, so as to minimize the average delay of
the vessels, where the delay of a vessel is defined to be the
difference between the departure time of the vessel (or equivalently
the time that the last cargo of the vessel has been reclaimed) and
the earliest time the vessel could depart under ideal circumstances,
i.e., the departure time if we assume the vessel arrives at its
$ETA$ and its stockpiles are ready to be reclaimed immediately upon
its arrival.

When assigning the cargoes of a vessel to locations in the
stockyard, scheduling their assembly, and scheduling their
reclaiming, the limited stockyard space, stacking rates, reclaiming
rates, and reclaimer movements have to be accounted for.

Since reclaimers are most likely to be the constraining entities in
the system, reclaimer activities need to be modeled at a fine level
of detail. That is all reclaimer activities, e.g., the reclaimer
movements along its rail track and the reclaiming of a stockpile,
have to be modeled in continuous time.

When deciding a stockpile location, a stockpile stacking start time,
and a stockpile reclaiming start time, a number of constraints have
to be taken into account: at any point in time no two stockpiles can
occupy the same space on a pad, reclaimers cannot be assigned to two
stockpiles at the same time, reclaimers can only be assigned to
stockpiles on pads that they serve, reclaimers serving the same pad
cannot pass each other, the stockpiles of a vessel have to be
reclaimed in a specified reclaim order and the time between the
reclaiming of consecutive stockpiles of a vessel can be no more than
a prespecified limit, the so-called continuous reclaim time limit,
and the reclaiming of the first stockpile of a vessel cannot start
before \textit{all} stockpiles of that vessel have been stacked. We focus on
some of the aspects of real world reclaimer scheduling as specified in Section 4.

The reclaiming of a stockpile using a bucket wheel reclaimer is
conducted in a series of long travel bench cuts. For each cut, the reclaimer moves along the whole
length of the stockpile with a fixed boom position. Then the boom is adjusted for the next cut, the
reclaimer turns around and moves along the stockpile in the opposite direction as indicated in
Figure~\ref{fig:reclaim}. The reclaiming process is fully automated and a typical stockpile is
reclaimed in three benches with approximately 55 cuts. In our simplified model we assume that a
stockpile is reclaimed while a reclaimer moves along it exactly once.

\begin{figure}[htb]
  \begin{minipage}[b]{.55\linewidth}
\centering
\begin{tikzpicture}
\tikzset{>=stealth'}
     \draw (0.8,0.4)--(-0.2,1)--(2,2.8);
     \draw (4.8,0.4)--(6,1)--(3.6,2.8);
     \draw (0.8,0.4)--(4.8,0.4);
     \draw (2,2.8)--(3.6,2.8);
      \draw[loosely dashed] (0.0,1.2)--(5.6,1.2);
      \draw[loosely dashed] (1,2)--(4.6,2);

      \draw [gray] (0.04,0.9332) arc (270:315:0.8);
      \draw [gray] (0.48,0.6664) arc (270:340:0.8);
      \draw [gray] (0.88,0.4) arc (270:360:0.8);
      \draw [gray] (1.28,0.4) arc (270:360:0.8);
      \draw [gray] (1.68,0.4) arc (270:360:0.8);
      \draw [gray] (2.08,0.4) arc (270:360:0.8);
      \draw [gray] (2.48,0.4) arc (270:360:0.8);
      \draw [gray] (2.88,0.4) arc (270:360:0.8);
      \draw [gray] (3.28,0.4) arc (270:360:0.8);
      \draw [gray] (3.68,0.4) arc (270:360:0.8);
      \draw [gray] (4.08,0.4) arc (270:360:0.8);
      \draw [gray] (4.48,0.4) arc (270:360:0.8);

      \draw [gray] (0.08,1.2) arc (270:349:1);
      \draw [gray] (0.56,1.2) arc (270:350:1);
      \draw [gray] (1,1.2) arc (270:350:1);
      \draw [gray] (1.4,1.2) arc (270:350:1);
      \draw [gray] (1.8,1.2) arc (270:350:1);
      \draw [gray] (2.2,1.2) arc (270:350:1);
      \draw [gray] (2.6,1.2) arc (270:350:1);
      \draw [gray] (3,1.2) arc (270:350:1);
      \draw [gray] (3.4,1.2) arc (270:350:1);
      \draw [gray] (3.8,1.2) arc (270:340:1.05);
      \draw [gray] (4.2,1.2) arc (270:330:1);
      \draw [gray] (4.6,1.2) arc (270:315:1);
      \draw [gray] (5,1.2) arc (270:300:1);

      \draw [thick] (1.04,2) arc (270:348:0.92);
      \draw [gray] (1.44,2) arc (270:352:0.92);
      \draw [gray] (1.84,2) arc (270:352:0.92);
      \draw [gray] (2.24,2) arc (270:352:0.92);
      \draw [gray] (2.64,2) arc (270:352:0.92);
      \draw [gray] (3.04,2) arc (270:339:0.92);
      \draw [gray] (3.44,2) arc (270:324:0.92);
      \draw [gray] (3.84,2) arc (270:308:0.92);

      \draw [->] (0.6,2.76) -- (1.70, 2.3);

     \node [ xshift=4.8cm, yshift = 2.6cm]{ {\small top bench }};
     \node [ xshift=6cm, yshift = 1.9cm]{ {\small middle bench }};
     \node [ xshift=6.4cm, yshift = 0.55cm]{ {\small bottom bench }};
     \node [ xshift=.8cm, yshift = 3.0cm]{ {\parbox{100pt}{$1^{\text{st}}$ cut top bench}}};
   \end{tikzpicture}
 \end{minipage}\hfill
\begin{minipage}[b]{.43\linewidth}
\centering
\begin{tikzpicture}[scale=3]
\tikzset{>=stealth'}
    \centering
     \draw (-0.05,0.0) rectangle  (0.55,1.2)  ;
      \draw[densely dashed, gray] (0.0,0.0) -- (0.0,1.2)  ;
      \draw[gray] (-0.02,0.025) -- (0.0,0) -- (0.02,0.025);
      \draw (0.05,0.0) -- (0.05,1.2)  ;
      \draw [gray] (0.0,-0.01) arc (180:360:0.05);
      \draw[gray] (0.08,-0.025) -- (0.1,0.0) -- (0.12,-0.025);
      \draw[densely dashed,  gray] (0.1,0.0) -- (0.1,1.2)  ;
      \draw[gray] (0.08,1.175) -- (0.1,1.2) -- (0.12,1.175);
      \draw (0.15,0.0) -- (0.15,1.2)  ;
      \draw [gray] (0.2,1.21) arc (0:180:0.05);
      \draw[gray] (0.18,1.225) -- (0.2,1.2) -- (0.22,1.225);
      \draw[densely dashed,  gray] (0.2,0.0) -- (0.2,1.2)  ;
      \draw[gray] (0.18,0.025) -- (0.2,0) -- (0.22,0.025);
      \draw (0.25,0.0) -- (0.25,1.2)  ;
      \draw [gray] (0.2,-0.01) arc (180:360:0.05);
      \draw[gray] (0.28,-0.025) -- (0.3,0.0) -- (0.32,-0.025);
      \draw[densely dashed,  gray] (0.3,0.0) -- (0.3,1.2)  ;
      \draw[gray] (0.28,1.175) -- (0.3,1.2) -- (0.32,1.175);
      \draw (0.35,0.0) -- (0.35,1.2)  ;
      \draw [gray] (0.4,1.21) arc (0:180:0.05);
      \draw[gray] (0.38,1.225) -- (0.4,1.2) -- (0.42,1.225);
      \draw[densely dashed,  gray] (0.4,0.0) -- (0.4,1.2)  ;
       \draw[gray] (0.38,0.025) -- (0.4,0) -- (0.42,0.025);
      \draw (0.45,0.0) -- (0.45,1.2)  ;
      \draw [gray] (0.4,-0.01) arc (180:360:0.05);
      \draw[gray] (0.48,-0.025) -- (0.5,0.0) -- (0.52,-0.025);
      \draw[densely dashed,  gray] (0.5,0.0) -- (0.5,1.2)  ;
      \draw[gray] (0.48,1.175) -- (0.5,1.2) -- (0.52,1.175);
      \draw[rounded corners,->] (0.9,0.45) -- (0.9,-0.04) -- (0.52,-0.04);
      \draw[rounded corners,->] (0.9,0.74) -- (0.9,1.24) -- (0.42,1.24);
      \node [ xshift=2.5cm, yshift = 1.8cm]{ {\parbox{20pt}{{\small turnaround points}}}};
\end{tikzpicture}
\end{minipage}
\caption{A cross section (left) and an aerial view (right) of a stockpile, illustrating the long travel bench cut.} \label{fig:reclaim}
\end{figure}
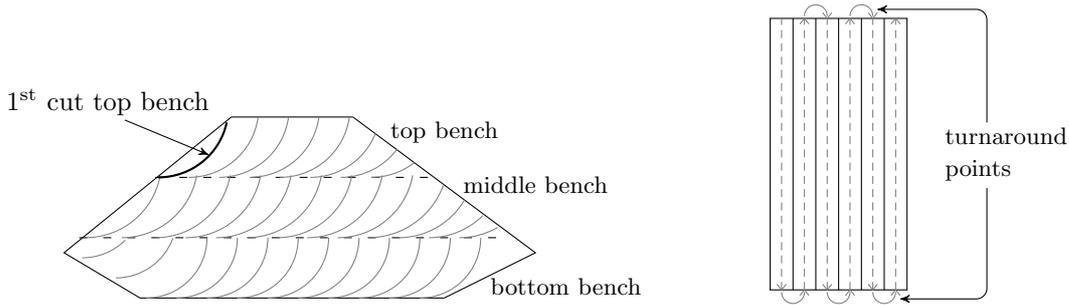

\section{Literature Review}
\label{sec:lit}

The scheduling of bucket wheel reclaimers in a coal terminal has
some similarities to the scheduling of quay and yard cranes in
container terminals. When a vessel arrives at a container terminal,
import containers are taken off the vessel and mounted onto trucks
by quay cranes and then unloaded by yard cranes at various locations
in the yard for storage. In the reverse operation, export containers
are loaded onto trucks by yard cranes at the yard, are off-loaded at
the quay, and loaded onto a vessel by quay cranes. Both reclaimers
and cranes move along a single rail track and therefore cannot pass
each other, and handle one object at a time (a stockpile in the case
of a bucket wheel reclaimer and a container in the case of quay or
yard cranes). Furthermore, in coal terminals as well as container
terminals maximizing throughput, i.e., minimizing the time it takes
to load and/or unload vessels, is the primary objective, and
achieving a high throughput depends strongly on effective scheduling
of the equipment.

However, there are also significant differences between the
scheduling of bucket wheel reclaimers in a coal terminal and the
scheduling of quay and yard cranes in the container terminals. The
number of containers that has to be unloaded from or loaded onto a
vessel in a container terminal is much larger than the number of
cargoes that has to be loaded onto a vessel in a coal terminal. As a
result, the sequencing of operations (e.g., respecting precedence
constraints between the unloading/loading of containers) is much
more challenging and a primary focus in crane scheduling problems.
On the other hand, containers and holds of a vessel has fixed length
whereas stockpiles can have an arbitrary length. As a consequence,
the time between the completion of one task and the start of the
next task (to account for any movement of equipment) can only take
on a limited number of values in crane scheduling problems at a
container terminal, especially in quay crane scheduling, and can
take on any value in reclaimer scheduling problem.

Below, we give a brief overview of the literature on scheduling of
quay and yard cranes in container terminals.

Most of the literature on quay crane scheduling focuses on a static
version of the problem in which the number of vessels, berth
assignments, and the quay crane assignments for each vessel are
known in advance for the entire planning horizon.

\cite{RePEc:eee:transb:v:23:y:1989:i:3:p:159-175} studies the
optimal assignment quay cranes (QCs) to the holds of multiple
vessels. In the considered setting, a task refers to the loading or
unloading of a single hold of a vessel and any precedence
constraints between the containers of a single hold are not
accounted for. Crane movement time between different holds is
assumed to be negligible. The fact that QCs cannot pass each other
is not explicitly considered. A mixed integer programming
formulation minimizing the (weighted) sum of departure times of the
vessels is presented. Some heuristics for a dynamic variant, in
which vessel arrival times are uncertain are also proposed.
\cite{peterkofsky1990branch} develop a branch-and-bound algorithm
for the same problem that is able to solve larger problem instances.

\cite{kim2004crane} studies the scheduling of multiple QCs
simultaneously operating on a single vessel taking into account
precedence constraints between containers and no-passing constraints
between QCs. In their setting, a task refers to a ``cluster'', where
a cluster represents a collection of adjacent slots in a hold and a
set of containers to be loaded into/unloaded from these slots. The
objective is to minimize a weighted combination of the load/unload
completion time of the vessel and the sum of the completion times of
the QCs, with higher weight for the load/unload completion time. The
sum of the completion times of the QCs is included in the objective
to ensure that QCs will be available to load/unload other vessels as
early as possible. The problem is modeled as an parallel machine
scheduling problem. A MIP formulation is presented and
branch-and-bound algorithm is developed for its solution. A greedy
randomized adaptive search procedure (GRASP) is developed to handle
instances where the branch-and-bound algorithm takes too much time.
\cite{NAV:NAV20121} have strengthened the MIP formulation of
\cite{kim2004crane} by deriving sets of valid inequalities. They
propose a branch-and-cut algorithm to solve the problem to
optimality.

\cite{doi:10.1080/03052150600691038} also study the scheduling of
multiple QCs for a single vessel with the objective of minimizing
the loading/unloading time. In their setting, a task refers to the
loading/unloading of a single hold, but {\it no-passing} constraints
for the QCs are explicitly taken into account. A heuristic that
partitions the holds of the vessel into non-overlapping zones and
assigns a QC to each zone is proposed. The optimal zonal partition
is found by dynamic programming. \cite{ZhuYandLiA2006} consider the
same setting, but formulate a mixed integer programming model and
propose a branch-and-bound algorithm for its solution (which
outperforms CPLEX on small instances). A simulated annealing
algorithm is developed to handle larger instances.

\cite{NAV:NAV20108} study a dynamic variant of the problem
considered by \cite{RePEc:eee:transb:v:23:y:1989:i:3:p:159-175},
which accounts for movement time of QCs and that QCs cannot pass
each other, and enforces a minimum separation between vessels. By
limiting the movement of QCs to be unidirectional, i.e., either from
``stern to bow'' or from ``bow to stern'' when loading/unloading a
vessel, it becomes relatively easy to handle the no-passing
constraint and to enforce a minimum separation. The objective is to
minimizing the time to load/unload multiple vessels. A heuristic is
proposed for the solution of the problem. \cite{NAV:NAV20189} study
a static, single vessel setting and also adopt a unidirectional
movement restriction for QCs to handle the no-passing constraint.
When a task refers to the loading/unloading of a hold, it is shown
that there always exists an optimal schedule among the
unidirectional schedules. Since a unidirectional schedule can be
easily obtained for a given task-to-QC assignment, a simulated
annealing algorithm is proposed to explore the space of task-to-QC
assignments.

For a more detailed, more comprehensive survey of crane scheduling,
the reader is referred to \cite{Bierwirth2010615}.

Recently \cite{legato2012modeling} presented a refined version of an
existing mixed integer programming formulation of the QC scheduling
problem incorporating many real-life constraints, such as QC service
rates, QC ready and due times, QC no-passing constraints, and
precedence constraints between groups of containers. Unidirectional
QC movements can be captured in the model as well. The best-known
branch-and-bound algorithm (i.e., from \cite{bierwirth2009fast}) is
improved with new lower bounding and branching techniques.

Yard crane (YC) scheduling is another critical component of the
efficient operation of a container terminal. The yard is typically
divided into several storage blocks and YCs are used to transfer
containers between these storage blocks and trucks (or prime
movers). YCs are either rail mounted or rubber wheeled. The rubber
wheeled YCs have the flexibility to move from one yard block to
another while rail mounted YCs are restricted to work on a single
yard block.

\cite{young1999routing} study the problem of minimizing the sum of
the set up times and the travel times of single YC. Their mixed
integer programming model determines the optimal route for the YC as
well as the containers to be picked up by the YC in each of the
storage blocks. Because of the excessive solve times of the mixed
integer program for large instances, two heuristics are proposed in
\cite{NAV:NAV10076}.

\cite{zhang2002dynamic} study the problem of scheduling a set of YCs
covering a number of storage blocks so as to minimize the total
tardiness (or delays). A mixed integer program model determines the
number of YCs to be deployed in each storage block in each planning
period and a lagrangian relaxation based heuristic algorithm is
employed to find an optimal solution.

\cite{Ng2005263, doi:10.1080/03052150500323849} studied the problem
of scheduling a YC that has to load/unload a given set of containers
with different ready times. The objective is to minimize the sum of
waiting times. The problem is formulated as mixed integer
programming problem, and a branch-and-bound algorithm is developed
for its solution. \cite{RePEc:eee:ejores:v:164:y:2005:i:1:p:64-78}
expands the study to the scheduling of multiple YCs in order to
minimize the total loading time or the sum of truck waiting times.
Because more than one YC can serve a storage block, a no-passing
constraint has to be enforced. A dynamic programming based heuristic
is proposed to solve the problem and a lower bound is derived to be
able to assess the quality of the solutions produced by the
heuristic.

\cite{petering2009effect} investigates how the width of the storage
blocks affects the efficiency of the operations at a container
terminal, given that the number of prime movers, the number of YCs,
the service rates of the YCs remain unchanged. A simulation study
indicates that the optimal storage block width ranges form 6 to 12
rows, depending on the size and shape of the terminal and the annual
number of containers handled by the terminal. Their experimental
results further show that restrictive YC mobility due to more
storage blocks gives better performance than a system with greater
YC mobility.

Only recently, researchers have started to examine the scheduling of
equipment in bulk goods terminals.
\cite{Hu:2012:SSD:2441300.2441305} consider the problem of
scheduling the stacker and reclaiming at a terminal for iron ore. It
is assumed that all tasks (stacking and reclaiming operation) are
known at the start of the planning horizon. The terminal
configuration is such that a single stracker/reclaimer serves two
pads, so there is no need to consider a no-passing constraint. A
sequence dependent set-up time, as a result of the movement of the
stacker/reclaimer between two consecutive tasks, is considered. A
mixed integer programming formulation is presented and a genetic
algorithm is proposed. \cite{6565191} study the problem of
scheduling reclaimers at an iron ore import terminal serving the
steel industry. Each reclaim task has a release date and due date,
and the goal is to minimize the completion of a set of reclaim
tasks. The terminal configuration is not specified and no mention is
made of no-passing constraints.  A mixed integer programming
formulation is presented and a Benders decomposition algorithms is
proposed for its solution.

For the single reclaimer case, the fact that the reclaimer has to travel along every stockpile gives it a traveling salesman
flavor. The basic traveling salesman problem is trivial when all nodes are on a line, but it becomes
difficult when additional constraints, such as time windows, are added (\cite{psaraftis1990routing},
\cite{tsitsiklis1992special}). In our problem, the single reclaimer case becomes a traveling salesman
problem on a line with prescribed edges: the nodes are the endpoints of the stockpiles, and for
every stockpile the edge connecting its endpoints has to be traversed in the solution.
   
\section{Problem Description}
\label{sec:problem}

The practical importance of reclaimer scheduling at a coal terminal
prompted us to study a set of simplified and idealized reclaimer
scheduling problems. These simplified and idealized reclaimer
scheduling problems turn out to lead to intriguing and, in some
cases, surprisingly challenging optimization problems.

We make the following basic assumptions:
\begin{itemize}
\item
There are two reclaimers $R_0$ and $R_1$ that serve two pads; one on
either side of the reclaimers.
\item
Reclaimer $R_0$ starts at one end of the stock pads and Reclaimer
$R_1$ starts at the other end of the stock pads.
\item  Stockpiles are reclaimed by one of the two reclaimers $R_0$ and $R_1$ that move forward and backward along a single rail in the aisle between the two pads.
\item
The reclaimers cannot pass each other but they can go along side by side.
\item
The reclaimers are identical, i.e., they have the same reclaim speed
and the same travel speed.
\item
Each stockpile has a given length and a given reclaim time (derived
from the stockpile's size and the reclaim speed of the reclaimers).
\item
When a stockpile is reclaimed, it has to be traversed along its
entire length by one of the reclaimers, either from left to right or
from right to left.
\item
After reclaiming the stockpiles, the reclaimers need to return to
their original position.
\end{itemize}

Using these basic assumptions, we define a number of variants of the
reclaimer scheduling problem:
\begin{itemize}
\item
Both reclaimers are used for the reclaiming of stockpiles or only
one reclaimer ($R_0$) is used to reclaim of stockpiles.
\item
The positions of the stockpiles on the pads are given or have to be
decided. If the positions on the pad are given, it is implicitly
assumed that the stockpile positions are feasible, i.e., that
stockpiles on the same pad do not overlap. If the positions have to
be decided, then both the pad and the location on the pad have to be
decided for each stockpile.
\item
Precedence constraints between stockpiles have to be observed or
not. When precedence constraints have to be observed, the reclaim
sequence of the stockpiles is completely specified. That is, the
precedence constraints form a chain involving all the stockpiles.
\end{itemize}

The goal in all settings is to reclaim all stockpiles and to
minimize the time at which both reclaimers have returned to their
original positions.

We use the following notation.  When the positions of the stockpiles
are given, we have two sets $J_1 = \{1, \ldots, n_1\}$ and $J_2 =
\{n_1+1, \ldots, n\}$ of stockpiles located on the two identical and
opposite pads $P_1$ and $P_2$. We represent a pad by segment
$[0,L]$, with $L$ being the length of the pad. Stockpile $j \in J_1$
occupies a segment $[l_j, r_j]$ on pad $P_1$ ($0 \leqslant l_j < r_j \leq
L$). Similarly, stockpile $j \in J_2$ occupies a segment $[l_j,
r_j]$ on pad $P_2$. Stockpiles cannot overlap on the same pad and we
assume that $r_j \leqslant l_{j+1}$ for $j \in \{1, \ldots, n_1-1\}$ and
for $j \in \{n_1+1, \ldots, n-1\}$ and that $r_j, l_j$ for $j \in
\{1, \ldots, n\}$ and $L$ are integers.

Reclaimers start and finish at the two endpoints of the rail,
reclaimer $R_0$ at point $0$ and reclaimer $R_1$ at point $L$, and
can reclaim stockpiles on either one of the pads (we assume there is
no time required to switch from one pad to the other), but they
cannot pass each other. A reclaimer can stay idle or move forward
and backward at the given speed $s$. When reclaiming a stockpile the
speed cannot be larger than $s$. Without loss of generality, we assume that the
reclaim speed is equal to 1 and the travel speed is $s\geqslant 1$. Thus, the time necessary to reclaim
stockpile $j$ is $p_j = r_j - l_j$, the length of stockpile $j$. 

When the positions of the stockpiles are not given but have to be
decided, we are given the length $p_j \in \Z$ of each stockpile $j$
($0 < p_j \leqslant L$) and we have to decide the pad on which to
locate the stockpile (either $P_1$ or $P_2$), the position $(l_j,
r_j)$ of that pad that the stockpile will occupy, and the reclaimer
schedules. 

\subsection{Graphical representation of a feasible schedule}

The schedule $H_k$ of reclaimer $R_k$ ($k=0,1$) with makespan $C_k$
can be described by a piecewise linear function representing the
position of the reclaimer on the rail as a function of time. Such a 
function can be represented by an ordered list of breakpoints
\[
B_k=\left(\left(t_i^{(k)},x^{(k)}_i\right)\in \R^{+}\times [0,L]\ :\ i=0,1,\ldots,q_k \right)
\]
in the time-space Cartesian plane, where $0=t_0^{(k)}<t_1^{(k)}<\cdots<t_{q_k}^{(k)}=C_k$,
$x_0^{(k)}=x_{q_k}^{(k)}=kL$. For $t\in[t_i^{(k)},t^{(k)}_{i+1}]$ we have
\[H_k(t)=x_i^{(k)}+\frac{x^{(k)}_{i+1}-x^{(k)}_{i}}{t^{(k)}_{i+1}-t^{(k)}_{i}}\left({t-t^{(k)}_{i}}\right),\]
and the slope between consecutive points $(t^{(k)}_i,x^{(k)}_i)$ and $(t^{(k)}_{i+1},x^{(k)}_{i+1})$ is either:
\begin{itemize}
\item $0$, the reclaimer is idle;
\item $+s$, the reclaimer is moving to the right without processing any
stockpile;
\item $-s$, the reclaimer is moving to left without processing any
stockpile;
\item $+1$, the reclaimer is moving to right while processing a stockpile on either one of the two
pads; and
\item $-1$, the reclaimer is moving to left while processing a stockpile on either one of the two
pads.
\end{itemize}
This is illustrated in Figure~\ref{fig:example}.
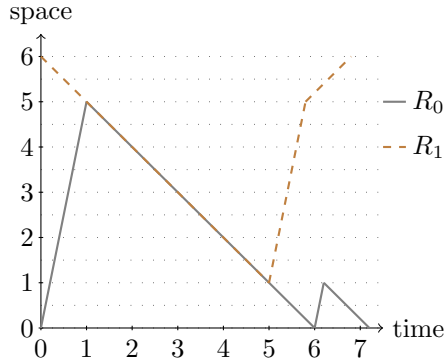
\begin{figure}[htb]
\centering
\begin{tikzpicture}[scale =0.3,
      declare function={
        func(\x) = (\x)/(4);
      }]
\pgfkeys{
     /pgf/number format/precision=1,
    /pgf/number format/fixed zerofill=true,
    /pgf/number format/fixed
}

    \draw[->] (0.0,0.0) -- (15.0,0.0) node[anchor =west] {{time}}  ;
    \draw[->] (0.0,0.0) -- (0.0,13.0) node[anchor = south] {{space}} ;
    \foreach \y in {0,...,6}
        {\draw(-.1,2*\y)--(.1,2*\y) node[left, xshift = -0.02] {{$\y$}};}
    \foreach \x in {0,...,7}
        {\draw (2*\x, 0.1) -- (2*\x, -0.1) node[below, yshift = -0.02] {{$\x$}};;
        }
    \foreach \y in {0,...,12}
        \draw[very thin, loosely dotted] (0.0,\y) -- (15,\y) ;

            \coordinate (SPLR) at (0.0,0.0);
            \coordinate (SPC) at ({(10.0-0.0)/5},10);
            \draw[gray, thick] (SPLR) -- (SPC);
            \coordinate (EPC) at ($({(10-0)/1}, 0) +( SPC |- 0,0.0)$);
            \draw[gray, thick] (SPC) -- (EPC);
            \coordinate (SPA) at ($({(2-0)/5},2) +( EPC |- 0,0)$);
            \draw[gray, thick] (EPC) -- (SPA);
            \coordinate (EPA) at ($({(2-0)/1}, 0) +( SPA |- 0,0.0)$);
            \draw[gray, thick] (EPA) -- (SPA);
            \draw[gray, thick] (15,10) -- (16,10);
            \node at (17,10) {$R_0$};
            \coordinate (SPRR) at (0.0,12);
            \coordinate (EPD) at ($({(12-2)/1},2) +( SPRR |- 0,0)$);
            \draw[brown, thick, dashed] (SPRR) -- (EPD);
            \coordinate (SPB) at ($({(10-2)/5}, 10) +( EPD |- 0,0.0)$);
            \draw[brown, thick, dashed] (EPD) -- (SPB);
            \coordinate (EPB) at ($({(12-10)/1}, 12) +( SPB |- 0,0)$);
            \draw[brown, thick, dashed] (SPB) -- (EPB);
            \draw[brown, thick, dashed] (15,8) -- (16,8);
            \node at (17,8) {$R_1$};
  \end{tikzpicture}
\caption{Reclaimer movement in time-space.}\label{fig:example}
\end{figure}
A pair ($H_0$, $H_1$) of reclaimer schedules is feasible if:
\begin{enumerate}
\item the two functions $H_0$ and $H_1$ satisfy the inequality
$H_1(t)\geqslant H_0(t), \forall t\geqslant 0$ (the reclaimers do not pass
each other);

\item each interval $[l_j,r_j]$ is traversed at least once at speed 1 (either from left to right or from right to
left); and

\item all other constraints are satisfied, e.g., precedence constraints between
stockpiles.
\end{enumerate}

The makespan of a feasible schedule $(H_0,H_1)$ is $C = \max(C_0, C_1)$. Next, we analyze a number of variants of the reclaimer schedule
problem. We start by considering variants in which the positions of
the stockpiles are given, which means only the schedules of the
reclaimers have to be determined. This is followed by considering
variants in which the positions of the stockpiles are not given, but
have to be determined, which means that both the stockpile positions
and the reclaimer schedules have to be determined.

\section{Reclaimer Scheduling Without Positioning Decisions}\label{sec:without}

\subsection{No precedence constraints}

\subsubsection{Single reclaimer}

For variants with a single reclaimer, we assume that the active
reclaimer is reclaimer $R_0$ with initial position $x_0 = 0$, and
that the schedule of reclaimer $R_1$ is $H_1(t) = L$ for $t
\geqslant 0$.

We start by observing that the optimal makespan $C^*$ cannot
be less than twice the time it takes to reach the farthest stockpile
endpoint $r= \max\{r_{n_1}, r_{n}\}$ at speed $s$ plus the
additional time to process the stockpiles, i.e.,
\begin{equation}\label{eq:lower_bound}
C^* \geqslant 2\frac{r}{s} + \sum_{j=1}^{n} \left[(r_j-l_j)-\frac{r_j-l_j}{s} \right]
\end{equation}
Next, we consider the \textbf{Forward-Backward (FB)} algorithm given in Algorithm~\ref{alg:fb},
where we omit the index $k$ for the reclaimer which is understood to be 0.
\begin{algorithm}[htb]
\caption{\textbf{Forward-Backward}} \label{alg:fb}
\begin{tabbing}
  ....\=....\=....\=............................... \kill \\
\textbf{Input:} $n_1$ and $ \{(l_j,r_j) | \ j=1,\ldots,n \} $ \\[1ex]
Initialize $q=0$ and $B=\left(\,(0,0)\,\right)$\\
\textbf{for} $j=1,\ldots,n_1$ \textbf{do}\\
\> \textbf{if} $l_j\neq x_{q}$ \textbf{then} add $\left(t_{q}+(l_j-x_{q})/s,\,l_j\right)$ to $B$ and increase $q$ by 1\\
\> Add $\left(t_{q}+(r_j-l_j),\,r_j\right)$ to $B$ and increase $q$ by 1\\
\textbf{for} $j=n,n-1,\ldots,n_1+1$ \textbf{do}\\
\> \textbf{if} $r_j\neq x_{q}$ \textbf{then} add $\left(t_{q}+\left\lvert r_j-x_{q}\right\rvert/s,\,r_j\right)$ to $B$ and increase $q$ by 1\\
\> Add $\left(t_{q}+(r_j-l_j),\,l_j\right)$ to $B$ and increase $q$ by 1\\
\textbf{if} $l_{n_1+1}\neq 0$ \textbf{then} add $\left(t_{q}+l_{n_1+1}/s,\,0\right)$ to $B$ and increase $q$ by 1\\[1ex]
\textbf{Output:} $C=t_{q}$ and $B$
\end{tabbing}
\end{algorithm}

\begin{thm}\label{thm:FB-optimal}
Algorithm~\ref{alg:fb} computes an optimal schedule for a single reclaimer in time $O(n)$.
\end{thm}
\begin{proof}
The schedule that is returned by Algorithm~\ref{alg:fb} is optimal
because by construction its makespan equals the lower bound given
by~(\ref{eq:lower_bound}). The algorithm runs in time $O(n)$.
\end{proof}

\subsubsection{Two reclaimers}

Unfortunately, optimally exploiting the additional flexibility and
extra opportunities offered by a second reclaimer is not easy as
we have the following theorem.

\begin{thm}\label{thm:np_hard_1}
Determining an optimal schedule for two reclaimers when the
positions of the stockpiles are given and the stockpiles can be
reclaimed in any order is NP-hard.
\end{thm}
\begin{proof}
We provide a transformation from \textsc{Partition}. An instance is
given by positive integers $a_1,\ldots,a_m$ and $B$ satisfying
$a_1+\cdots+a_m=2B$ and the problem is to decide if there is an
index set $I\subseteq\{1,\ldots,m\}$ with $\sum_{i\in I}a_i=B$. We
reduce this to the following instance of the reclaimer scheduling
problem. The length of the pad is $L=6B$, the travel speed is
$s=5B$, and we have $n=m+2$ stockpiles which are all placed on pad
$P_1$, i.e., $n_1=n$. The stockpile lengths  are $a_i$ ($i = 1,
\ldots, m$) for the first $m$ stockpiles and the two additional
stockpiles have both length $2B$. The positions of the stockpiles on
the pad are given by $(l_{m+1}, r_{m+1}) = (0,2B)$, $(l_{m+2},
r_{m+2}) = (4B,6B)$, and $(l_i, r_i) = (2B + \sum_{j=1}^{i-1} a_i,
2B + \sum_{j=1}^i a_i)$ for $i = 1, \ldots, m$. We claim that a
makespan $\leqslant 3B+1$ can be achieved if and only if the
\textsc{Partition} instance is a YES-instance. Clearly, if there is
no $I$ with $\sum_{i\in I}a_i=B$, we cannot divide the stockpiles
between the two reclaimers in such a way that the total stockpile
length for both reclaimers is $3B$, which implies that one of the
reclaimers has a reclaim time of at least $3B+1$, hence its makespan
is larger than $3B+1$ (as the reclaimer also has to travel without
reclaiming a stockpile). Conversely, if there is an $I$ with
$\sum_{i\in I}a_i=B$, we can achieve a makespan of less than or
equal to $3B+1$ as follows. Reclaimer $R_0$ moves from $x=0$ to
$x=4B$ while reclaiming (from left to right) stockpile $m+1$ and the
stockpiles with index in $I$, and then it returns to its start point
at time $3B+1$. Reclaimer $R_1$ moves from $x=L$ to $x=2B$ without
reclaiming anything, and then it moves back to $x=L$, reclaiming
(from left to right) the stockpiles with index in
$\{1,\ldots,m\}\setminus I$ and stockpile $m+2$. There is no
clashing because the region that is visited by both reclaimers is
the interval $[2B,4B]$, and reclaimer $R_0$ enters this interval at
time $2B$ and leaves it at time $2B+3/5$, while reclaimer $R_1$
enters at time $2/5$ and leaves at time $B+1$.
\end{proof}

To be able to analyze the quality of schedules for two reclaimers,
we start by deriving a lower bound. For this purpose, we allow
preemption, i.e., we allow a stockpile to be split and be processed
either simultaneously or at different times by any of the two
reclaimers. Let
\begin{align}
  S_1 &= \bigcup\limits_{j=1}^{n_1}[l_j,r_j]&&\text{ and}& S_2 &= \bigcup\limits_{j=n_1+1}^{n}[l_j,r_j]\label{eq:stockpiles}
\end{align}
be the subsets of $[0,L]$ that represent occupied space on pads
$P_1$ and $P_2$, respectively. Furthermore, let
\begin{align*}
Q_1 &= S_1 \triangle S_2,& Q_1 &= S_1 \cap S_2,&& \text{and}& E &=
[0,L] \setminus (S_1 \cup S_2)
\end{align*}
be the subsets of $[0,L]$ with stockpiles on one side, with a
stockpile on both sides, and with no stockpile on either side,
respectively. Note that $E$ is a union of finitely many pairwise
disjoint intervals, say
\begin{equation}\label{eq:empty_area}
E = [a_1, b_1] \cup [a_2, b_2] \cup \cdots \cup [a_r, b_r].
\end{equation}

For a subset $X \subseteq[0,L]$ let $\ell(X)$ denote the (total)
length of $X$. The set $Q_2$ has to be traversed twice with speed 1
and the set $Q_1$ has to be traversed once with speed 1 and once
with traveling speed $s$, hence $C_1+C_2\geqslant 2\ell(Q_2) +
\ell(Q_1)+ \ell(Q_1)/s$, which implies the lower bound
\[C=\max\{C_0,\,C_1\}\geqslant\frac12\left[2\ell(Q_2) + \ell(Q_1)+ \ell(Q_1)/s\right].\]

We can improve this bound by taking into account the set $E$. Note
that at most one of the intervals in the
partition~(\ref{eq:empty_area}) can contain points that are not
visited by any reclaimer, because otherwise the stockpiles between
two such points are not reclaimed. Now we consider two cases.
\begin{description}
\item[Case 1.]
If every point of the set $E$ is visited, then the set $E$ is
traversed (at least) twice with speed $s$, so in this case the
makespan $C$ is at least
\[K_0 = \frac12 \left[ 2\ell(Q_2) + \ell(Q_1) + \ell(Q_1)/s + 2\ell(E)/s \right].\]

\item[Case 2.]
If some point in the interval $[a_i,b_i]$, $i \in \{1, \ldots, r\}$,
is not visited, then everything left of $a_i$ is reclaimed by $R_0$,
while everything right of $b_i$ is reclaimed by $R_1$, so in this
case the makespan $C$ is at least
\[K_i = \max \left\{ 2\ell(Q_2^-) + \ell(Q_1^-) +\frac{\ell(Q_1^-)}{s} + 2\frac{\ell(E^-)}{s},\ 2\ell(Q_2^+) + \ell(Q_1^+) +\frac{\ell(Q_1^+)}{s} + 2\frac{\ell(E^+)}{s} \right\} ,\]
where $Q_2^- = Q_2 \cap[0, a_i]$, $Q_2^+ = Q_2 \cap [b_i, L]$, and
similarly for $Q_1$ and $E$.
\end{description}
\begin{thm}\label{thm:preempt_opt}
The optimal makespan for a preemptive schedule equals
\[K^*=\min\{K_i\ :\ i=0,1\ldots,r\},\]
and an optimal schedule can be computed in linear time.
\end{thm}
\begin{proof}
By the above discussion, $K^*$ is a lower bound for the makespan of
a preemptive schedule. We define two functions $f,g:[0,L]\to\R$ as
follows. Let $f(x)$ be the return time of reclaimer $R_0$ if it
moves from $0$ to $x$, reclaiming everything on this part of pad
$P_1$, and then moves back to $0$ while reclaiming everything on
this part of pad $P_2$. Similarly, let $g(x)$ be the return time of
reclaimer $R_1$ if it moves from $L$ to $x$ reclaiming everything on
this part of pad $P_1$, and then moves back to $L$, reclaiming
everything on this part of pad $P_2$. These are piecewise linear,
continuous functions, which can be expressed in terms of the sets
$Q_1$, $Q_2$ and $E$:
\begin{align*}
  f(x) &= 2\ell(Q_2^-(x)) + \ell(Q_1^-(x)) +\frac{\ell(Q_1^-(x))}{s} + 2\frac{\ell(E^-(x))}{s},\\
  g(x) &= 2\ell(Q_2^+(x)) + \ell(Q_1^+(x)) +\frac{\ell(Q_1^+(x))}{s} + 2\frac{\ell(E^+(x))}{s},
\end{align*}
where $Q_2^-(x) = Q_2 \cap[0, x]$, $Q_2^+(x) = Q_2 \cap [x, L]$, and
similarly for $Q_1$ and $E$. Note that
$K_i=\max\{f(a_i),\,g(b_i)\}$. The functions $f$ and $g$ satisfy the
following conditions:
\begin{itemize}
\item $f(0)=g(L)=0$ and $f(L)=g(0)=2K_0$,
\item $f(x)+g(x)=2K_0$ for all $x\in[0,L]$, and
\item $f$ is strictly increasing, and $g$ is strictly decreasing.
\end{itemize}
This implies that there is a unique $x^*\in[0,L]$ with
$f(x^*)=g(x^*)=K_0$.
\begin{description}
\item[Case 1.] There is at least one stockpile at position $x^*$, i.e., $x^*\in Q_1\cup Q_2$. In this case, for any interval $[a_i,b_i]$ in the partition~(\ref{eq:empty_area}), either $b_i\leqslant x^*$ or $a_i\geqslant x^*$, hence $K_i=\max\{f(a_i),\,g(b_i)\}\geqslant K_0$, and consequently $K^*=K_0$. This value is achieved by reclaiming everything left of $x^*$ by reclaimer $R_0$ and everything right of $x^*$ by reclaimer $R_1$ as described in the definition of the functions $f$ and $g$.
\item[Case 2.] There is no stockpile at position $x^*$, i.e., $x^*\in[a_i,b_i]$ for some interval $[a_i,b_i]$ in the partition~(\ref{eq:empty_area}). Then $K_i=\max\{f(a_i),\,g(b_i)\}\leqslant f(x^*)=K_0$. For $j<i$, we have $K_j\geqslant g(b_j)>g(x^*)=K_0$, and for $j>i$, $K_j\geqslant f(a_j)>f(x^*)=K_0$. Hence $K^*=K_i$, and this value is achieved by reclaiming everything left of $a_i$ by reclaimer $R_0$ and everything right of $b_i$ by reclaimer $R_1$ as described in the definition of the functions $f$ and $g$.
\end{description}
This concludes the proof of the optimality of the value $K^*$. In
order to compute $x^*$, which defines an optimal schedule, we order
the numbers $0,l_1,r_1,\ldots,l_n,r_n,L$ increasingly, which can be
done in linear time, because we assume that the stockpiles on each
pad are already ordered from left to right. This gives an ordered
list
\[0=x_0\leqslant x_1\leqslant x_2\leqslant\cdots\leqslant x_{2n+2}=L\]
of the breakpoints of the piecewise linear functions $f$ and $g$. We
can determine the values of $f$ and $g$ at these points recursively,
by $f(x_0)=0$,
\[f(x_{k})=
  \begin{cases}
    f(x_{k-1})+(x_{k}-x_{k-1})\cdot 2/s & \text{if }[x_{k-1},\,x_k]\subseteq E\\
    f(x_{k-1})+(x_{k}-x_{k-1})\left(1+1/s\right) & \text{if }[x_{k-1},\,x_k]\subseteq Q_1\\
    f(x_{k-1})+(x_{k}-x_{k-1})\cdot 2 & \text{if }[x_{k-1},\,x_k]\subseteq Q_2
  \end{cases}
\]
for $k=1,2,\ldots,2n+2$, and $g(x_k)=2K_0-f(x_k)$. Then there is a
unique index $k$ with $f(x_{k-1})\leqslant K_0<f(x_k)$, and we
obtain $x^*$ by
\[x^*=x_{k-1}+\frac{K_0-f(x_{k-1})}{f(x_k)-f(x_{k-1})}\cdot(x_k-x_{k-1}).\qedhere\]
\end{proof}

In order to describe and analyze non-preemptive schedules, we
introduce some additional notation and a few more functions. For
$x\in[0,L]$, the region occupied by stockpiles left (resp. right) of
$x$ on pad $i$ is denoted by $S^-_i(x)$ (resp. $S^+_i(x)$). More
precisely, with $S_1$ and $S_2$ defined by~(\ref{eq:stockpiles}),
\begin{align*}
S^-_i(x)&=S_i\cap[0,x], & S^+_i(x)&=S_i\cap[x,L].
\end{align*}
Furthermore, we define functions $f_i,g_i:[0,L]\to\R$ for
$i\in\{1,2\}$ by
\begin{align*}
  f_i(x) &= \ell(S^-_i(x))+\frac{x-\ell(S^-_i(x))}{s},&
  g_i(x) &= \ell(S^+_i(x))+\frac{L-x-\ell(S^+_i(x))}{s}.
\end{align*}
Note that $f(x)=f_1(x)+f_2(x)$ and $g(x)=g_1(x)+g_2(x)$, where $f$
and $g$ are the functions defined in the proof of
Theorem~\ref{thm:preempt_opt}.

Let $j\in\{1,\ldots,n_1\}$ be a stockpile on pad $P_1$, and let
$j'\in\{n_1+1,\ldots,n\}$ be a stockpile on pad $P_2$. If $R_0$
reclaims all stockpiles left of (and including) $j$ on pad $P_1$,
and all stockpiles left of (and including) $j'$ on pad $P_2$, then
its earliest possible return time is
\begin{equation}\label{eq:F_unimod}
F(j,j')=f_1(r_j)+\lvert r_j-r_{j'}\rvert/s+f_2(r_{j'}).
\end{equation}
Similarly, if $R_1$ reclaims all stockpiles right of (not including)
$j$ on pad $P_1$, and all stockpiles right of (not including) $j'$,
then its earliest possible return time is
\begin{equation}\label{eq:G_unimod}
G(j,j')=g_1(l_{j+1})+\lvert l_{j+1}-l_{j'+1}\rvert/s+g_2(l_{j'+1}).
\end{equation}
See Figure~\ref{fig:unimodal} for an illustration.
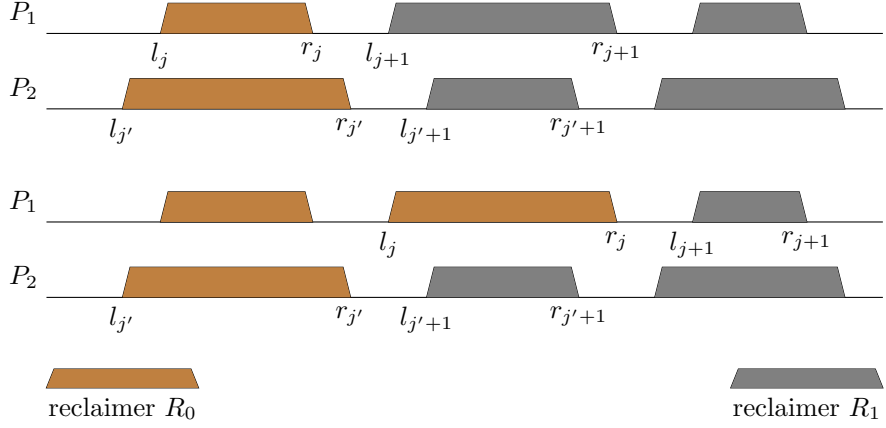
\begin{figure}[htb]
  \centering
  \begin{tikzpicture}
  \draw (0,0) -- (11,0);
  \draw (0,1) -- (11,1);
  \draw  (1.5,1) -- (1.6,1.4) -- (3.4,1.4) -- (3.5,1) -- cycle;
  \fill [brown]  (1.5,1) -- (1.6,1.4) -- (3.4,1.4) -- (3.5,1) -- cycle;
  \draw  (4.5,1) -- (4.6,1.4) -- (7.4,1.4) -- (7.5,1) -- cycle;
  \fill [brown]  (4.5,1) -- (4.6,1.4) -- (7.4,1.4) -- (7.5,1) -- cycle;
  \coordinate [label = below:$l_{j}$] (l11) at (4.5,1);
  \coordinate [label = below:$r_{j}$] (r11) at (7.5,1);
  \draw  (8.5,1) -- (8.6,1.4) -- (9.9,1.4) -- (10,1) -- cycle;
  \fill [gray]  (8.5,1) -- (8.6,1.4) -- (9.9,1.4) -- (10,1) -- cycle;
  \coordinate [label = below:$l_{j+1}$] (l12) at (8.5,1);
  \coordinate [label = below:$r_{j+1}$] (r12) at (10,1);
  \draw  (1,0) -- (1.1,.4) -- (3.9,.4) -- (4,0) -- cycle;
  \fill [brown]  (1,0) -- (1.1,.4) -- (3.9,.4) -- (4,0) -- cycle;
  \coordinate [label = below:$l_{j'}$] (l21) at (1,0);
  \coordinate [label = below:$r_{j'}$] (r21) at (4,0);
  \draw  (5,0) -- (5.1,.4) -- (6.9,.4) -- (7,0) -- cycle;
  \fill [gray]  (5,0) -- (5.1,.4) -- (6.9,.4) -- (7,0) -- cycle;
  \coordinate [label = below:$l_{j'+1}$] (l22) at (5,0);
  \coordinate [label = below:$r_{j'+1}$] (r22) at (7,0);
  \draw  (8,0) -- (8.1,.4) -- (10.4,.4) -- (10.5,0) -- cycle;
  \fill [gray]  (8,0) -- (8.1,.4) -- (10.4,.4) -- (10.5,0) -- cycle;
  \draw (0,2.5) -- (11,2.5);
  \draw (0,3.5) -- (11,3.5);
  \draw  (1.5,3.5) -- (1.6,3.9) -- (3.4,3.9) -- (3.5,3.5) -- cycle;
  \fill [brown]  (1.5,3.5) -- (1.6,3.9) -- (3.4,3.9) -- (3.5,3.5) -- cycle;
  \draw  (4.5,3.5) -- (4.6,3.9) -- (7.4,3.9) -- (7.5,3.5) -- cycle;
  \fill [gray]  (4.5,3.5) -- (4.6,3.9) -- (7.4,3.9) -- (7.5,3.5) -- cycle;
  \coordinate [label = below:$l_{j}$] (l11a) at (1.5,3.5);
  \coordinate [label = below:$r_{j}$] (r11a) at (3.5,3.5);
  \draw  (8.5,3.5) -- (8.6,3.9) -- (9.9,3.9) -- (10,3.5) -- cycle;
  \fill [gray]  (8.5,3.5) -- (8.6,3.9) -- (9.9,3.9) -- (10,3.5) -- cycle;
  \coordinate [label = below:$l_{j+1}$] (l12a) at (4.5,3.5);
  \coordinate [label = below:$r_{j+1}$] (r12a) at (7.5,3.5);
  \draw  (1,2.5) -- (1.1,2.9) -- (3.9,2.9) -- (4,2.5) -- cycle;
  \fill [brown]  (1,2.5) -- (1.1,2.9) -- (3.9,2.9) -- (4,2.5) -- cycle;
  \coordinate [label = below:$l_{j'}$] (l21a) at (1,2.5);
  \coordinate [label = below:$r_{j'}$] (r21a) at (4,2.5);
  \draw  (5,2.5) -- (5.1,2.9) -- (6.9,2.9) -- (7,2.5) -- cycle;
  \fill [gray]  (5,2.5) -- (5.1,2.9) -- (6.9,2.9) -- (7,2.5) -- cycle;
  \coordinate [label = below:$l_{j'+1}$] (l22a) at (5,2.5);
  \coordinate [label = below:$r_{j'+1}$] (r22a) at (7,2.5);
  \draw  (8,2.5) -- (8.1,2.9) -- (10.4,2.9) -- (10.5,2.5) -- cycle;
  \fill [gray]  (8,2.5) -- (8.1,2.9) -- (10.4,2.9) -- (10.5,2.5) -- cycle;
  \draw (0,-1.2) -- (0.1,-.95) -- (1.9,-.95) -- (2,-1.2) -- cycle;
  \fill[brown] (0,-1.2) -- (0.1,-.95) -- (1.9,-.95) -- (2,-1.2) -- cycle;
  \coordinate [label = below:reclaimer $R_0$] (R0) at (1,-1.2);
  \draw (9,-1.2) -- (9.1,-.95) -- (10.9,-.95) -- (11,-1.2) -- cycle;
  \fill[gray] (9,-1.2) -- (9.1,-.95) -- (10.9,-.95) -- (11,-1.2) -- cycle;
  \coordinate [label = below:reclaimer $R_1$] (R1) at (10,-1.2);
  \coordinate [label = above:$P_1$] (P1a) at (-.3,1);
  \coordinate [label = above:$P_2$] (P2a) at (-.3,0);
  \coordinate [label = above:$P_1$] (P1b) at (-.3,3.5);
  \coordinate [label = above:$P_2$] (P2b) at (-.3,2.5);
  \end{tikzpicture}
\caption{Two stockpile assignments for non-preemptive schedules (indicated by the colors of the stockpiles): at
the top with $\min\{l_{j+1},\,l_{j'+1}\}\geqslant\max\{r_j,r_{j'}\}$
and at the bottom with $l_{j'+1}<r_j$.}\label{fig:unimodal}
\end{figure}
Observe that if $l_{j'+1}\geqslant r_j$ and $l_{j+1}\geqslant
r_{j'}$, then no clashes will occur between the two reclaimers and a
makespan of $C(j,j') = \max(F(j,j'),\,G(j,j'))$ can be achieved. On
the other hand, if $l_{j'+1}<r_j$ or $l_{j+1}<r_{j'}$, then it can
happen that one reclaimer has to wait. Therefore, in order to
specify a schedule, we have to
\begin{itemize}
\item
Choose one of two options for the routing of $R_0$: (1) first
reclaim stockpiles $1,2,\ldots,j$ on pad $P_1$ from left to right
and then stockpiles $j',j'-1,\ldots,n_1+1$ on pad $P_2$ from right
to left, or (2) first reclaim stockpiles $n_1+1,\ldots,j'$ on pad
$P_2$ from left to right and then stockpiles $j,j-1,\ldots,1$ on pad
$P_1$ from right to left;
\item
Choose one of two options for the routing of $R_1$: (1) first
reclaim stockpiles $n_1,n_1-1,\ldots,j+1$ on pad $P_1$ from right to
left and then stockpiles $j'+1,\ldots,n$ on pad $P_2$ from left to
right, or (2) first reclaim stockpiles $n,n-1,\ldots,j'+1$ on pad
$P_2$ from right to left and then stockpiles $j+1,\ldots,n_1$ on pad
$P_1$ from left to right; and
\item
Choose which reclaimer waits.
\end{itemize}
Taking all possible combinations we have $8$ different schedules for
a given pair $(j,j')$ of stockpiles. For $p,q\in\{1,2\}$ and
$k\in\{0,1\}$, let $C_{pqk}$ be the makespan that results from
routing option $p$ for $R_0$, routing option $q$ for $R_1$, and
letting $R_k$ wait if necessary. We describe the computation of
$C_{pqk}$ in detail for $l_{j'+1}<r_j$ and $k=1$. The cases with
$l_{j+1}<r_{j'}$ or $k=0$ can be treated in the same way. Since
$R_1$ waits if necessary, the makespan of $R_0$ is $F(j,j')$,
defined in~(\ref{eq:F_unimod}). So
$C_{pq1}=\max\{F(j,j'),C^1_{pq1}\}$, where $C^1_{pq1}$ is the
corresponding makespan for $R_1$ and can be computed as follows. In
each case we express $C^1_{pq1}$ as $G(j,j')+w$ where $G(j,j')$ is
the lower bound for the makespan of $R_1$ given
in~(\ref{eq:G_unimod}) and $w$ is the waiting time which is the
expression in square brackets in the equations below.
\begin{description}
\item[Case 1.]
Both reclaimers start on pad $P_1$. If $g_1(r_j)\geqslant f_1(r_j)$,
then no waiting is necessary and $C^1_{111}=G(j,j')$. Otherwise
$R_1$ waits at $x=r_j$ until $R_0$ arrives there at time $f_1(r_j)$,
hence
\[C^1_{111} = g_1(r_j) + [f_1(r_j)-g_1(r_j)] +(r_j-l_{j'+1})/s+g_2(l_{j'}).\]

\item[Case 2.]
$R_0$ starts on pad $P_1$ and $R_1$ starts on pad $P_2$. If
$g_2(l_{j'+1})\leqslant f_1(l_{j'+1})$, then no waiting is necessary
and $C^1_{121}=G(j,j')$. Otherwise $R_1$ waits on its way to
$x=l_{j'+1}$ for a period of length $f_1(r_j)-g_2(r_j)$ and the
makespan is
\[C^1_{121} = g_2(l_{j'+1})+[f_1(r_j)-g_2(r_j)]+(l_{j+1}-l_{j'+1})/s+g_1(l_{j+1}).\]

\item[Case 3.]
$R_0$ starts on pad $P_2$ and $R_1$ starts on pad $P_1$. If
$g_1(r_j)\geqslant f_2(r_{j'})+(r_j-r_{j'})/s$, then $R_1$ arrives
at $x=r_j$ when $R_0$ is already on its way back, no waiting is
necessary, and $C^1_{211}=G(j,j')$. Otherwise $R_1$ waits at
$x=l_{j+1}$ for a period of length
$f_2(r_{j'})+(r_j-r_{j'})/s-g_1(r_j)$ and the makespan is
\[C^1_{211} = g_1(l_{j+1})+[f_2(r_{j'})+(r_j-r_{j'})/s-g_1(r_j)]+(l_{j+1}-l_{j'+1})/s+g_2(l_{j'+1}).\]

\item[Case 4.]
Both reclaimers start on pad $P_2$. If $g_2(l_{j'+1})\leqslant
f_2(l_{j'+1})$, then $R_1$ is already on its way back when $R_0$
arrives at $x=l_{j'+1}$, no waiting is necessary, and
$C^1_{211}=G(j,j')$. Otherwise $R_1$ waits to the right of $x=r_j$
for a period of length
$f_2(r_{j'})+(r_j-r_{j'})/s+f_1(r_j)-f_1(l_{j'+1})$ to arrive at
$x=l_{j'+1}$ at the same time as $R_0$ on its way back, and the
makespan is
\[C^1_{221} = g_2(l_{j'+1})+[f_2(r_{j'})+(r_j-r_{j'})/s+f_1(r_j)-f_1(l_{j'+1})]+(l_{j+1}-l_{j'+1})/s+g_1(l_{j+1}).\]
\end{description}
The necessary data to evaluate the $8$ schedules associated with a
pair $(j,j')$ can be computed in linear time in the same way as the
functions $f$ and $g$ are evaluated in the proof of
Theorem~\ref{thm:preempt_opt}.

In the discussion above, we have assumed that a schedule has the
following properties:
\begin{itemize}
\item Each reclaimer changes direction exactly once.
\item
Reclaimer $R_0$ reclaims everything between $0$ and some point on
one pad from left to right, and then everything from some (possibly
different) point to 0 on the other pad from right to left.
\item
Reclaimer $R_1$ reclaims everything between $L$ and some point on
one pad from right to left, and then everything from some (possibly
different) point to $L$ on the other pad from left to right.
\end{itemize}
In the following, we call such a schedule \emph{contiguous
unimodal}. Since every contiguous unimodal schedule is associated
with some stockpile pair, we have the following theorem.
\begin{thm}\label{thm:best_unimodal}
An optimal contiguous unimodal schedule can be computed in quadratic
time. \qed
\end{thm}

The next example shows that it is possible that there is no optimal
contiguous unimodal schedule.
\begin{example}\label{ex:unimod_not_opt}
Consider an instance with four stockpiles of lengths
$2,\,10,\,10,\,2$ shown in Figure~\ref{fig:instance}, and let the
travel speed be $s=5$.
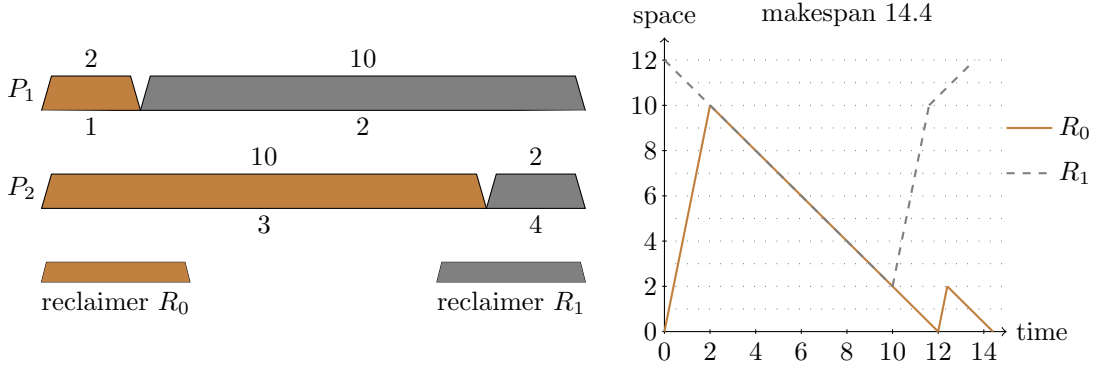
\begin{figure}[htb]
\centering
    \begin{minipage}{.4\linewidth}
    \hspace*{-3em}\raisebox{-.9em}{
        \begin{tikzpicture} [scale =0.65]
        \coordinate (A) at (0,0);
        \coordinate (B) at (0.2,0.7);
        \coordinate (C) at (1.8,0.7);
        \coordinate (D) at (2,0);
        \fill [brown] (A) -- (B) -- (C) -- (D) -- cycle;
        \fill [gray] (2,0) -- (2.2,0.7) -- (10.8,0.7) -- (11.00,0) -- cycle;
        \draw (A) -- (B) -- (C) -- (D) -- cycle;
        \draw  (2,0) -- (2.2,0.7) -- (10.8,0.7) -- (11.00,0) -- cycle;
        \coordinate [label = above:$1$] (A) at (1.0,-.7);
        \coordinate [label = above:$2$] (B) at (6.5,-.7);
        \coordinate [label = above:$2$] (A) at (1.0,0.7);
        \coordinate [label = above:$10$] (B) at (6.5,0.7);
        \fill [white] (0,0) -- (0,-0.1) -- (11.00,0) -- (11.00,-0.1) -- cycle;
        \fill [brown] (0,-2) -- (0.2,-1.3) -- (8.8,-1.3) -- (9.00,-2) -- cycle;
        \fill [gray] (9.0,-2) -- (9.2,-1.3) -- (10.8,-1.3) -- (11.00,-2) -- cycle;
        \draw (0,-2) -- (0.2,-1.3) -- (8.8,-1.3) -- (9.00,-2) -- cycle;
        \draw (9.00,-2) -- (9.2,-1.3) -- (10.8,-1.3) -- (11.00,-2) -- cycle;
        \coordinate [label = above:$3$] (A) at (4.5,-2.7);
        \coordinate [label = above:$4$] (B) at (10.0,-2.7);
        \coordinate [label = above:$10$] (G) at (4.5,-1.3);
        \coordinate [label = above:$2$] (A) at (10.0,-1.3);
          \draw (0,-3.5) -- (0.1,-3.1) -- (2.9,-3.1) -- (3,-3.5) -- cycle;
  \fill[brown] (0,-3.5) -- (0.1,-3.1) -- (2.9,-3.1) -- (3,-3.5) -- cycle;
  \coordinate [label = below:reclaimer $R_0$] (R0) at (1.5,-3.5);
  \draw (8,-3.5) -- (8.1,-3.1) -- (10.9,-3.1) -- (11,-3.5) -- cycle;
  \fill[gray] (8,-3.5) -- (8.1,-3.1) -- (10.9,-3.1) -- (11,-3.5) -- cycle;
  \coordinate [label = below:reclaimer $R_1$] (R1) at (9.5,-3.5);
  \coordinate [label = above:$P_1$] (P1) at (-.4,0);
  \coordinate [label = above:$P_2$] (P2) at (-.4,-2);       
  \end{tikzpicture}
        }
    \end{minipage}
    \begin{minipage}{.4\linewidth}
    \hspace*{3em}\raisebox{-0.3em}{
        \begin{tikzpicture}[scale = 0.3,
            declare function={
                func(\x) = (4.0)*(\x);
            }]
        \pgfkeys{
            /pgf/number format/precision=1,
            /pgf/number format/fixed zerofill=true,
            /pgf/number format/fixed
        }
            \node at (8,14) {{makespan $14.4$}};
            \draw[->] (0.0,0.0) -- (15.0,0.0) node[anchor =west] {{time}}  ;
            \draw[->] (0.0,0.0) -- (0.0,13.0) node[anchor = south] {{space}} ;
            \foreach \y in {0,...,6}
                { \pgfmathtruncatemacro{\result}{2*\y}
                                  \draw(-.1,2*\y)--(.1,2*\y) node[left, xshift = -0.02] {{\result}};}
            \foreach \x in {0,...,7}
                { \pgfmathtruncatemacro{\result}{2*\x}
                                  \draw (2*\x, 0.1) -- (2*\x, -0.1) node[below, yshift = -0.02] {{\result}};;
                }
            \foreach \y in {0,...,12}
                \draw[very thin, loosely dotted] (0.0,\y) -- (15,\y) ;
            \coordinate (SPLR) at (0.0,0.0);
            \coordinate (SPC) at ({(10.0-0.0)/5},10);
            \draw[brown, thick] (SPLR) -- (SPC);
            \coordinate (EPC) at ($({(10-0)/1}, 0) +( SPC |- 0,0.0)$);
            \draw[brown, thick] (SPC) -- (EPC);
            \coordinate (SPA) at ($({(2-0)/5},2) +( EPC |- 0,0)$);
            \draw[brown, thick] (EPC) -- (SPA);
            \coordinate (EPA) at ($({(2-0)/1}, 0) +( SPA |- 0,0.0)$);
            \draw[brown, thick] (EPA) -- (SPA);
            \draw[brown, thick] (15,9) -- (17,9);
            \node at (18,9) {{$R_0$}};
            \coordinate (SPRR) at (0.0,12);
            \coordinate (EPD) at ($({(12-2)/1},2) +( SPRR |- 0,0)$);
            \draw[gray, thick, dashed] (SPRR) -- (EPD);
            \coordinate (SPB) at ($({(10-2)/5}, 10) +( EPD |- 0,0.0)$);
            \draw[gray, thick, dashed] (EPD) -- (SPB);
            \coordinate (EPB) at ($({(12-10)/1}, 12) +( SPB |- 0,0)$);
            \draw[gray, thick, dashed] (SPB) -- (EPB);
            \draw[gray, thick, dashed] (15,7) -- (17,7);
            \node at (18,7) {{$R_1$}};
        \end{tikzpicture}
        }
    \end{minipage}
\caption{Instance demonstrating that unimodal routing is not always
optimal. The assignment between stockpiles and reclaimers is indicated by different colors, the
numbers above the stockpiles are the stockpile lengths while the numbers below are the stockpile indices.}\label{fig:instance}
\end{figure}
Unimodal routing results in $C = 15.2$. However, when $R_0$
first travels to $x=10$ without processing any stockpile, then
processes stockpile 3 while coming back, then travels to $x=2$, and
finally processes stockpile 1, and $R_1$ first processes stockpile
2, then turns and processes stockpile 4 on the way back, the
resulting makespan is 14.4. This shows that sometimes ``zigzagging''
can be beneficial.
\end{example}

Next, we analyze one particular schedule, which is obtained by a
natural modification of the optimal preemptive schedule and
therefore can be computed in linear time.

Let $x^*$ be the optimal split point for a preemptive schedule as
described in Theorem~\ref{thm:preempt_opt}, and let $k$ be the index
with $x_{k-1}<x^*\leqslant x_k$. If $x^*\in E$, i.e., there is no
stockpile at $x^*$, then the optimal preemptive schedule is actually
non-preemptive, and yields an optimal solution with makespan $K^*$,
the optimal preemptive makespan. In general, the stockpiles are
assigned according to the following rules (see
Figure~\ref{fig:assign}).
\begin{enumerate}
\item All stockpiles $j$ with $r_j\leqslant x^*$ are assigned to $R_0$, and all stockpiles $j$ with $l_j\geqslant x^*$ are assigned to $R_1$.
\item If there is exactly one stockpile $j$ with $l_j\leqslant x^*\leqslant r_j$ then this stockpile is assigned to $R_0$ if $x^*-l_j\geqslant r_j-x^*$ and to $R_1$ otherwise.
\item If there are two stockpiles $j\in\{1,\ldots,n_1\}$ and $j'\in\{n_1+1,\ldots,n\}$ with $l_j\leqslant x^*\leqslant r_j$ and $l_{j'}\leqslant x^*\leqslant r_{j'}$ then both of them are assigned to $R_0$ if
  \begin{equation}\label{eq:left_assign}
(x^*-l_j)+(x^*-l_{j'})+\lvert
l_j-l_{j'}\rvert/s\geqslant(r_j-x^*)+(r_{j'}-x^*)+\lvert
r_j-r_{j'}\rvert/s,
  \end{equation}
and otherwise both stockpiles are assigned to $R_1$.
\end{enumerate}
\begin{figure}[htb]
  \centering
  \begin{tikzpicture}
  \draw (0,2.5) -- (11,2.5);
  \draw (0,3.5) -- (11,3.5);
  \draw  (.5,3.5) -- (.6,3.9) -- (2.4,3.9) -- (2.5,3.5) -- cycle;
  \fill [brown]  (.5,3.5) -- (.6,3.9) -- (2.4,3.9) -- (2.5,3.5) -- cycle;
  \draw  (3.5,3.5) -- (3.6,3.9) -- (7.9,3.9) -- (8,3.5) -- cycle;
  \fill [brown]  (3.5,3.5) -- (3.6,3.9) -- (7.9,3.9) -- (8,3.5) -- cycle;
  \coordinate [label = below:$l_{j}$] (l11) at (3.5,3.5);
  \coordinate [label = below:$r_{j}$] (r11) at (8,3.5);
  \draw  (9.5,3.5) -- (9.6,3.9) -- (10.4,3.9) -- (10.5,3.5) -- cycle;
  \fill [gray]  (9.5,3.5) -- (9.6,3.9) -- (10.4,3.9) -- (10.5,3.5) -- cycle;
  \draw[dashed] (6,2.5) -- (6,4.1);
  \draw[dashed] (4.5,2.5) -- (4.5,4.1);
  \draw[dashed] (7,2.5) -- (7,4.1);
  \coordinate [label = right:$x^*$] (x) at (6,4.1);
  \coordinate [label = right:$x_{k-1}$] (x1) at (4.5,4.1);
  \coordinate [label = right:$x_k$] (x2) at (7,4.1);
  \draw  (1,2.5) -- (1.1,2.9) -- (1.4,2.9) -- (1.5,2.5) -- cycle;
  \fill [brown]  (1,2.5) -- (1.1,2.9) -- (1.4,2.9) -- (1.5,2.5) -- cycle;
  \draw  (3,2.5) -- (3.1,2.9) -- (4.4,2.9) -- (4.5,2.5) -- cycle;
  \fill [brown]  (3,2.5) -- (3.1,2.9) -- (4.4,2.9) -- (4.5,2.5) -- cycle;
  \draw  (8,2.5) -- (8.1,2.9) -- (10.4,2.9) -- (10.5,2.5) -- cycle;
  \fill [gray]  (8,2.5) -- (8.1,2.9) -- (10.4,2.9) -- (10.5,2.5) -- cycle;
  \draw  (7,2.5) -- (7.1,2.9) -- (7.4,2.9) -- (7.5,2.5) -- cycle;
  \fill [gray]  (7,2.5) -- (7.1,2.9) -- (7.4,2.9) -- (7.5,2.5) -- cycle;
  \draw (0,0) -- (11,0);
  \draw (0,1) -- (11,1);
  \draw  (1.5,1) -- (1.6,1.4) -- (3.4,1.4) -- (3.5,1) -- cycle;
  \fill [brown]  (1.5,1) -- (1.6,1.4) -- (3.4,1.4) -- (3.5,1) -- cycle;
  \draw  (4.5,1) -- (4.6,1.4) -- (8.9,1.4) -- (9,1) -- cycle;
  \fill [gray]  (4.5,1) -- (4.6,1.4) -- (8.9,1.4) -- (9,1) -- cycle;
  \draw  (10,1) -- (10.1,1.4) -- (10.4,1.4) -- (10.5,1) -- cycle;
  \fill [gray]  (10,1) -- (10.1,1.4) -- (10.4,1.4) -- (10.5,1) -- cycle;
  \coordinate [label = below:{$l_{j}$}] (l12a) at (4.35,1);
  \coordinate [label = below:$r_{j}$] (r12a) at (9,1);
  \draw  (.5,0) -- (.6,.4) -- (1.4,.4) -- (1.5,0) -- cycle;
  \fill [brown]  (.5,0) -- (.6,.4) -- (1.4,.4) -- (1.5,0) -- cycle;
  \draw  (5,0) -- (5.1,.4) -- (6.9,.4) -- (7,0) -- cycle;
  \fill [brown]  (2.5,0) -- (2.6,.4) -- (6.9,.4) -- (7,0) -- cycle;
  \coordinate [label = below:$l_{j'}$] (l22a) at (2.5,0);
  \coordinate [label = below:{$r_{j'}$}] (l22a) at (7,0);
  \draw  (10,0) -- (10.1,.4) -- (10.4,.4) -- (10.5,0) -- cycle;
  \fill [gray]  (10,0) -- (10.1,.4) -- (10.4,.4) -- (10.5,0) -- cycle;
  \draw  (7.5,0) -- (7.6,.4) -- (7.9,.4) -- (8,0) -- cycle;
  \fill [gray]  (7.5,0) -- (7.6,.4) -- (7.9,.4) -- (8,0) -- cycle;
  \draw  (8.5,0) -- (8.6,.4) -- (9.4,.4) -- (9.5,0) -- cycle;
  \fill [gray]  (8.5,0) -- (8.6,.4) -- (9.4,.4) -- (9.5,0) -- cycle;
  \coordinate [label = above:$P_1$] (P1a) at (-.3,1);
  \coordinate [label = above:$P_2$] (P2a) at (-.3,0);
  \coordinate [label = above:$P_1$] (P1b) at (-.3,3.5);
  \coordinate [label = above:$P_2$] (P2b) at (-.3,2.5);
  \draw[dashed] (6,0) -- (6,1.6); \coordinate [label = right:$x^*$] (x) at (6,1.6);
  \draw[dashed] (4.5,0) -- (4.5,1.6); \coordinate [label = right:$x_{k-1}$] (x1) at (4.5,1.6);
  \draw[dashed] (7,0) -- (7,1.6); \coordinate [label = right:$x_k$] (x2) at (7,1.6);
  \draw (0,-1.2) -- (0.1,-.95) -- (1.9,-.95) -- (2,-1.2) -- cycle;
  \fill[brown] (0,-1.2) -- (0.1,-.95) -- (1.9,-.95) -- (2,-1.2) -- cycle;
  \coordinate [label = below:reclaimer $R_0$] (R0) at (1,-1.2);
  \draw (9,-1.2) -- (9.1,-.95) -- (10.9,-.95) -- (11,-1.2) -- cycle;
  \fill[gray] (9,-1.2) -- (9.1,-.9) -- (10.9,-.9) -- (11,-1.2) -- cycle;
  \coordinate [label = below:reclaimer $R_1$] (R1) at (10,-1.2);
  \end{tikzpicture}
\caption{Possible assignments (indicated by colors) for contiguous unimodal schedules:
$x^*\in S_1\setminus S_2$ (top), $x^*\in S_1\cap S_2$ (bottom).}
  \label{fig:assign}
\end{figure}
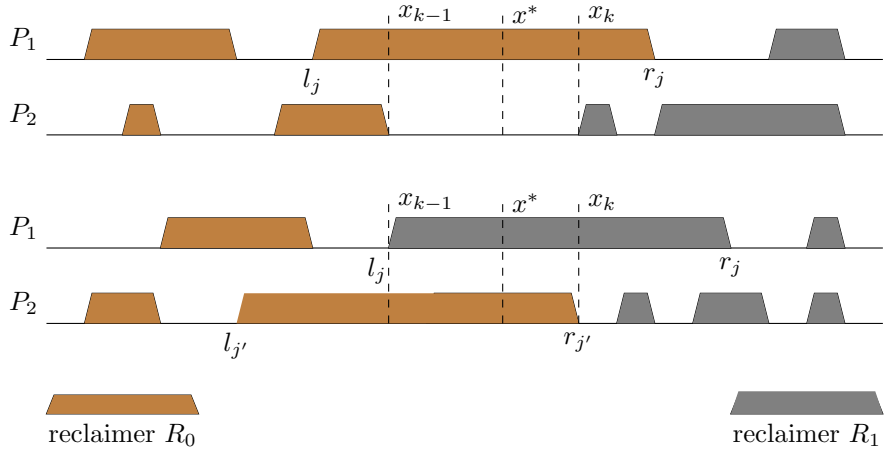
Let $(\tilde H_0,\,\tilde H_1)$ be the best unimodal schedule
associated with this stockpile assignment, and let $\tilde
C=\max\{\tilde C_0,\,\tilde C_1\}$ be the makespan of this schedule.
\begin{thm}\label{thm:2_approx}
We have $\tilde C\leqslant 2K^*$. In particular, the schedule
$(\tilde H_0,\,\tilde H_1)$ provides a 2-approximation for the
problem of scheduling two reclaimers when the positions of the
stockpiles are given and the stockpiles can be reclaimed in any
order. The factor 2 is asymptotically best possible (for
$s\to\infty$).
\end{thm}
\begin{proof}
To see that the factor of $2$ cannot be improved, consider two
stockpiles each of length $L$. According to our rule, both
stockpiles are assigned to the left reclaimer and this yields a
makespan of $\tilde C=2L$. On the other hand, assigning one
stockpile to each reclaimer and reclaiming both of them from right
to left yields a makespan of $C^*=(1+2/s)L$.

Without loss of generality, we make the following assumptions.
\begin{itemize}
\item
If $x^*\in S_1\triangle S_2$, then $x^*\in S_1\setminus S_2$ and for
the stockpile $j$ with $l_j<x^*<r_j$, we have $x^*-l_j\geqslant
r_j-x^*$, so that stockpile $j$ is assigned to $R_0$.
\item
If $x^*\in S_1\cap S_2$, then for the two stockpiles
$j\in\{1,\ldots,n_1\}$ and $j'\in\{n_1+1,\ldots,n\}$ with
$l_j<x^*<r_j$ and $l_{j'}<x^*<r_{j'}$, we have $r_j\geqslant r_{j'}$
and (\ref{eq:left_assign}) holds, so that both stockpiles are
assigned to $R_0$.
\end{itemize}
Furthermore, it turns out that in order to establish the factor 2
bound it is sufficient to consider the setting where $R_0$ starts on
pad $P_1$, $R_1$ starts on pad $P_2$, and $R_1$ waits if necessary.
We start by bounding $\tilde C_0$, the makespan for $R_0$. If
$x^*\in S_1\setminus S_2$ then
\[\tilde C_0=f_1(r_j)+(r_j-x^*)/s+f_2(x^*)=f_1(x^*)+(r_j-x^*)(1+1/s)+f_2(x^*)\leqslant 2K^*,\]
where the last inequality follows from
\[K^*=f_1(x^*)+f_2(x^*)\geqslant (x^*-l_j)+x^*/s\geqslant(r_j-x^*)(1+1/s).\]
If $x^*\in S_1\cap S_2$ then
\begin{multline*}
\tilde C_0=f_1(x^*)+(r_j-x^*)+(r_j-r_{j'})/s+f_2(x^*)+(r_{j'}-x^*)\\
= K^*+(r_j-x^*)+(r_j-r_{j'})/s+(r_{j'}-x^*)\leqslant 2K^*.
\end{multline*}
The bound for the makespan $\tilde C_1$ of $R_1$ can be derived
simultaneously for both cases after noting that in our setting we
have $r_{j'}=x_k$. If $g_2(x_k)\leqslant f_1(x_k)$ or
$g_2(r_j)\geqslant f_1(r_j)$ then no waiting is necessary, and the
makespan of $R_1$ is
\[\tilde C_1\leqslant g_2(x_k)+(r_j-x_k)+g_1(r_j)\leqslant g_1(x_k)+g_2(x_k)\leqslant K^*.\]
Otherwise the waiting time for $R_1$ is
\[f_1(r_j)-g_2(r_j)=[f_1(x_k)+(r_j-x_k)]-g_2(r_j)<g_2(x_k)-g_2(r_j)+(r_j-x_k)\]
hence the makespan is (using $g_1(x_k)=g_1(r_j)+(r_j-x_k)$)
\begin{multline*}
\tilde C_1= g_2(x_k)+[g_2(x_k)-g_2(r_j)+(r_j-x_k)]+(r_j-x_k)/s+g_1(r_j)\\
= [g_1(x_k)+g_2(x_k)]+[g_2(x_k)-g_2(r_j)+(r_j-x_k)/s]\leqslant
2K^*.\qedhere
\end{multline*}
\end{proof}

Example~\ref{ex:unimod_not_opt} shows that unimodal routing might
not be optimal. Next, we examine whether contiguous assignment is
always optimal, i.e., whether there always exists an optimal
schedule characterized by two stockpiles $j \in \{1,\ldots,n_1\}$ on
pad $P_1$ and $j' \in \{n_1+1,\ldots,n\}$ on pad $P_2$ and an
associated assignment of stockpiles $\{1,\ldots,j,
j',j'-1,\ldots,n_1+1\}$ to $R_0$ and the remaining stockpiles to
$R_1$. We call such a schedule a \emph{contiguous} schedule. The
next example shows that it is possible that no contiguous schedule
is optimal.
\begin{example}\label{ex:weak_uni_not_opt}
Consider the instance (illustrated in
Figure~\ref{fig:Counter_Assign}) with $n=n_1=5$, i.e., all
stockpiles on pad $P_1$, and stockpile positions $(0,1) \, (1,2),
\,(2,4), \, (4,5), \, (5,6)\}$. The best we can do with a contiguous
schedule is to assign stockpiles $1$ and $2$ to $R_0$ and the other
stockpiles to $R_1$, which yields a makespan of $4+4/s$. But by
assigning stockpiles $1$, $2$ and $4$ to $R_0$ and stockpiles $3$
and $5$ to $R_1$ we can achieve a makespan of $3+9/s$. The ratio
$\frac{4+4/s}{3+9/s}$ tends to $4/3$ for $s\to\infty$.
\begin{figure}[htb]
\centering
\begin{tikzpicture}[scale =0.55,
      declare function={
        func(\x) = (4.0)*(\x);
      }]
\pgfkeys{
     /pgf/number format/precision=1,
    /pgf/number format/fixed zerofill=true,
    /pgf/number format/fixed
}
      \draw[thick] (-3.2,-1.3) rectangle (7.6,8.6);
     \node at (2.2,8) {{\small contiguous: makespan $4.22$}};
      \draw[->] (-1,0.0) -- (5,0.0) node[anchor = west] {{\small time}} ;
      \draw[->] (-1,0.0) -- (-1,7) node[anchor = east] {{\small space}} ;
\foreach \y in {1,2,3,4,5,6} {
    \draw (-0.98,\y) -- ( -1.02,\y) node[left, xshift = -0.02] {{\small $\y$}};
} \foreach \x in {0,1,2,3,4,5} {
    \draw (\x-1.0, 0.02) -- (\x-1.0, -0.02) node[below, yshift = -0.02]  {{\small $\x$}};
}

    \draw[very thin, loosely dotted] (-1,1) -- (5,1) ;
    \draw[very thin, loosely dotted] (-1,2) -- (5,2) ;
    \draw[very thin, loosely dotted] (-1,3) -- (5,3) ;
    \draw[very thin, loosely dotted] (-1,4) -- (5,4) ;
    \draw[very thin, loosely dotted] (-1,5) -- (5,5) ;
    \draw[very thin, loosely dotted] (-1,6) -- (5,6) ;
      \coordinate (SPLRC) at (-1,0.0);
      \coordinate (SPAA) at (-1,0);
      \draw[brown, thick] (SPLRC) -- (SPAA);
      \coordinate (EPAA) at ($({(1-0)/1}, 1) +( SPAA |- 0,0)$);
      \draw[brown, thick] (SPAA) -- (EPAA);
      \coordinate (SPBB) at ($({(1-1)/18}, 1) +( EPAA |- 0,0.0)$);
      \draw[brown, thick] (EPAA) -- (SPBB);
      \coordinate (EPBB) at ($({(2-1)/1},2) +( SPBB |- 0,0)$);
      \draw[brown, thick] (SPBB) -- (EPBB);
      \coordinate (EPLRC) at ($({(2-0.0)/18}, 0.0) +( EPBB |- 0,0)$);
      \draw[brown, thick] (EPBB) -- (EPLRC);
      \draw[brown, thick] (5.5,4) -- (6.1,4);
      \node at (6.7,4) {{\small $R_0$}};
      \coordinate (SPRRC) at (-1.0,6);
      \coordinate (SPEE) at (-1,6);
      \draw[gray, thick, dashed] (SPRRC) -- (SPEE);
      \coordinate (EPEE) at ($({(6-5)/1}, 5) +( SPEE |- 0,0)$);
      \draw[gray, thick, dashed] (SPEE) -- (EPEE);
      \coordinate (SPDD) at ($({(5-5)/18}, 5) +( EPEE |- 0,0.0)$);
      \draw[gray, thick, dashed] (EPEE) -- (SPDD);
      \coordinate (EPDD) at ($({(5-4)/1}, 4) +( SPDD |- 0,0)$);
      \draw[gray, thick, dashed] (SPDD) -- (EPDD);
      \coordinate (EPCC) at ($({(4-2)/1}, 2) +( EPDD |- 0,0)$);
      \draw[gray, thick, dashed] (EPDD) -- (EPCC);
      \coordinate (EPRRC) at ($({(6-2)/18}, 6) +( EPCC |- 0,0)$);
      \draw[gray, thick, dashed] (EPCC) -- (EPRRC);
      \draw[gray, thick, dashed] (5.5,3) -- (6.1,3);
      \node at (6.7,3) {{\small $R_1$}};


      \draw[thick] (8.5,-1.3) rectangle (19.3,8.6);
     \node at (13.9,8) {{\small optimal: makespan 3.5}};
      \draw[->] (10.7,0.0) -- ({10.7+6},0.0) node[anchor = west] {{\small time}} ;
      \draw[->] (10.7,0.0) -- (10.7,7) node[anchor = east] {{\small space}} ;
\foreach \y in {1,2,3,4,5,6} {
    \draw (10.72,\y) -- ( 10.68,\y) node[left, xshift = -0.02] {{\small $\y$}};
} \foreach \x in {10,11,12,13,14,15} {
    \pgfmathtruncatemacro{\result}{\x-10}
    \draw (\x+0.7, 0.02) -- (\x+0.7, -0.02) node[below, yshift = -0.02] {{\small \result}};
}
    \draw[very thin, loosely dotted] (10.7,1) -- (16.7,1) ;
    \draw[very thin, loosely dotted] (10.7,2) -- (16.7,2) ;
    \draw[very thin, loosely dotted] (10.7,3) -- (16.7,3) ;
    \draw[very thin, loosely dotted] (10.7,4) -- (16.7,4) ;
    \draw[very thin, loosely dotted] (10.7,5) -- (16.7,5);
    \draw[very thin, loosely dotted] (10.7,6) -- (16.7,6) ;
      \coordinate (SPLR) at (10.7,0.0);
      \draw[brown, thick] (SPLR) -- ({10.7+(1)/(9)},0);
      \coordinate (EPA) at ({12.7+(1)/(9)},2);
      \draw[brown, thick] ({10.7+(1)/(9)},0) -- (EPA);
      \coordinate (EPB) at ({12.7+(2)/(9)},4);
      \draw[brown, thick] (EPA) -- (EPB);
      \coordinate (SPD) at ({13.7+(2)/(9)},5);
      \draw[brown, thick] (EPB) -- (SPD);
      \coordinate (EPD) at ({10.7+3.5},0);
      \draw[brown, thick] (SPD) -- (EPD);
      \draw[brown, thick] (17.2,4) -- (17.8,4);
      \node at (18.4,4) {{\small $R_0$}};
      \coordinate (SPRR) at (10.7,6);
      \coordinate (SPC) at ($({(6-2)/18},2) +(10.7,0)$);
      \draw[gray, thick, dashed] (SPRR) -- (SPC);
      \coordinate (EPC) at ($({(4-2)/1}, 4) +( SPC |- 0,0)$);
      \draw[gray, thick, dashed] (SPC) -- (EPC);
      \coordinate (SPE) at ($({(5-4)/18},5) +( EPC |- 0,0.0)$);
      \draw[gray, thick, dashed] (EPC) -- (SPE);
      \coordinate (EPE) at ($({(6-5)/1}, 6) +( SPE |- 0,0)$);
      \draw[gray, thick, dashed] (SPE) -- (EPE);
      \draw[gray, thick, dashed] (17.2,3) -- (17.8,3);
      \node at (18.4,3) {{\small $R_1$}};
  \end{tikzpicture}
\bigskip

\begin{tikzpicture}[scale=1.8]
 \fill [gray] (0,0) -- (0.2,0.3) -- (0.8,0.3) -- (1,0) -- cycle;
 \fill [gray] (1,0) -- (1.2,0.3) -- (1.8,0.3) -- (2,0) -- cycle;
 \fill [gray] (2,0) -- (2.2,0.3) -- (3.8,0.3) -- (4,0) -- cycle;
 \fill [gray] (4,0) -- (4.2,0.3) -- (4.8,0.3) -- (5,0) -- cycle;
 \fill [gray] (5,0) -- (5.2,0.3) -- (5.8,0.3) -- (6,0) -- cycle;
 \draw  (0,0) -- (0.2,0.3) -- (0.8,0.3) -- (1,0) -- cycle;
 \draw  (1,0) -- (1.2,0.3) -- (1.8,0.3) -- (2,0) -- cycle;
 \draw  (2,0) -- (2.2,0.3) -- (3.8,0.3) -- (4,0) -- cycle;
 \draw  (4,0) -- (4.2,0.3) -- (4.8,0.3) -- (5,0) -- cycle;
 \draw  (5,0) -- (5.2,0.3) -- (5.8,0.3) -- (6,0) -- cycle;
 \draw (0,0) -- (6,0);
 \coordinate [label = below:{\small $1$}] (A) at (0.5,0);
 \coordinate [label = below:{\small $2$}] (B) at (1.5,0);
 \coordinate [label = below:{\small $3$}] (C) at (3,0);
 \coordinate [label = below:{\small $4$}] (D) at (4.5,0);
 \coordinate [label = below:{\small $5$}] (E) at (5.5,0);
\coordinate [label = above:$P_1$] (P1a) at (-.3,0);
\coordinate [label = above:$P_2$] (P2a) at (-.3,-.6);
\draw (0,-.6) -- (6.0,-.6);
\end{tikzpicture}
\caption{An instance where contiguous scheduling is not optimal
(with $s=18$). At the top the reclaimer movements and at the bottom
the stockpile positions.}\label{fig:Counter_Assign}
\end{figure}
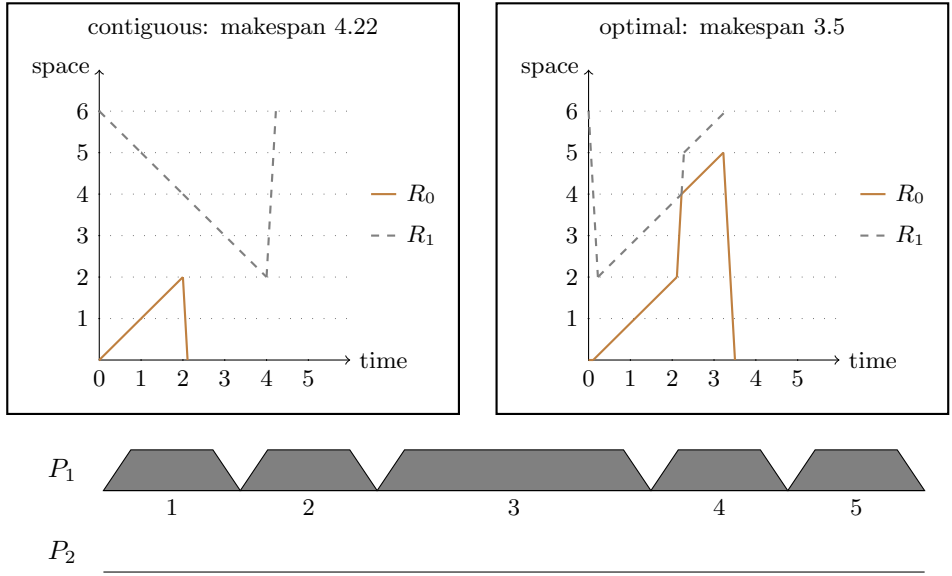
\end{example}

We conjecture that Example~\ref{ex:weak_uni_not_opt} represents the
worst case for the performance of contiguous schedules.
\begin{cnj}\label{con:weak_uni_ratio}
An optimal contiguous schedule provides a $4/3$-approximation for
the problem of scheduling two reclaimers when the positions of the
stockpiles are given and the stockpiles can be reclaimed in any
order.
\end{cnj}

A natural approach for finding an optimal contiguous schedule is to
determine an optimal schedule for each pair
$(j,j')\in\{1,\ldots,n_1\} \times \{n_1+1,\ldots,n\}$ and pick the
best one. Unfortunately, it is already an NP-hard problem to
determine the best routing for a given contiguous assignment of stockpiles to
reclaimers.
\begin{thm}\label{thm:oracle_hardness}
\red{Suppose that we are given the stockpile positions and a contiguous assignment of the stockpiles to
two reclaimers. Then it is NP-hard to find an optimal schedule.}
\end{thm}
\begin{proof}
We provide a reduction from \textsc{Partition}. Let the instance be
given by positive integers $a_1,\ldots,a_m$ whose sum is $2B$. A
corresponding instance of the reclaimer scheduling problem is
constructed as follows (see Figure~\ref{fig:instance-pf} for an
illustration). The pad length is $L=53B$ and the travel speed is
$s=2$. On pad $P_1$, there are $m+3$ stockpiles with lengths $a_1$,
$\ldots$, $a_m$, $9B$, $26B$ and $16B$, and on pad $P_2$ there are
$3$ stockpiles with lengths $35B$, $6B$ and $2B$. The positions of
the first $m$ stockpiles corresponding to the integers from the
\textsc{Partition} instance are determined by
\begin{align*}
  l_j &= L-\sum_{i=1}^ja_i, & r_j &= l_j+a_j && \text{for }j\in\{1,\ldots,m\}
\end{align*}
and the positions of the 6 dummy stockpiles are $[0,9B]$,
$[9B,35B]$, $[35B,51B]$, $[0,35B]$, $[35B,41B]$ and $[41B,43B]$. Let
stockpiles $1,\ldots, m, m+2, m+3$, and $m+6$ be assigned to
reclaimer $R_1$ and let stockpiles $m+1$, $m+4$, and $m+5$ be
assigned to reclaimer $R_0$.
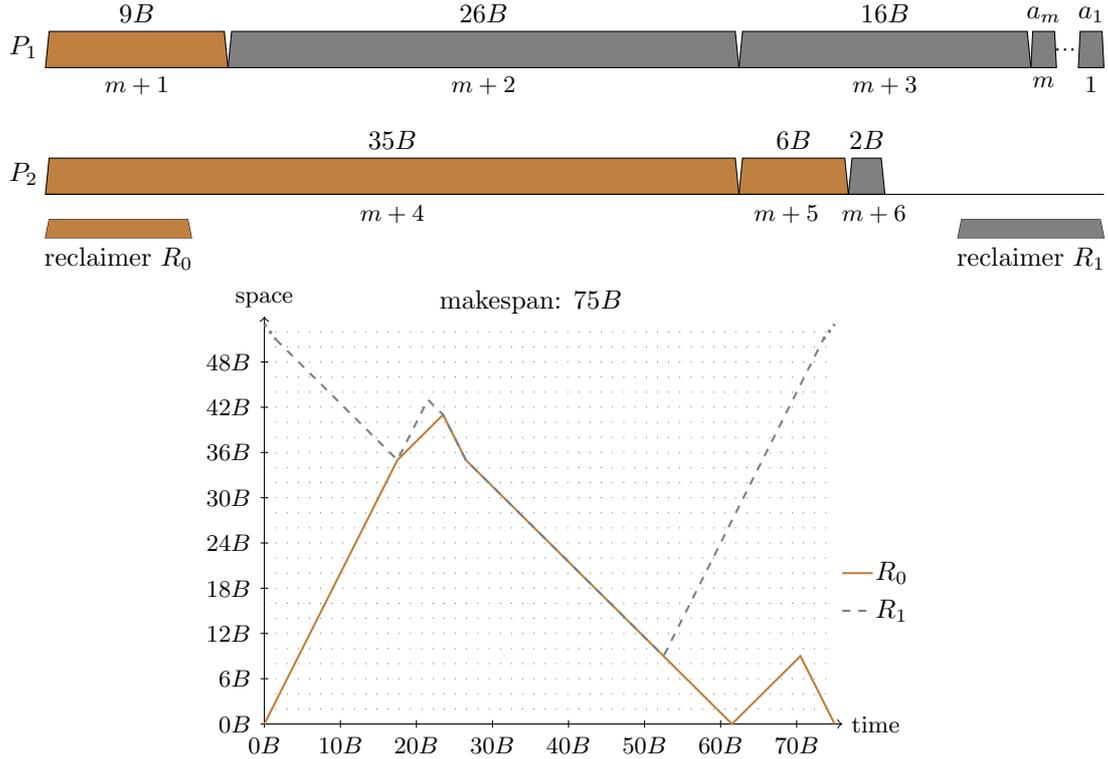
\begin{figure}[htb]
\centering
 \begin{tikzpicture}[scale=0.48]
        \coordinate (A) at (0,0);
        \coordinate (B) at (0.1,1);
        \coordinate (C) at (4.9,1);
        \coordinate (D) at (5,0);
        \fill [brown] (A) -- (B) -- (C) -- (D) -- cycle;
        \fill [gray] (5,0) -- (5.1,1) -- (18.9,1) -- (19,0) -- cycle;
        \fill [gray] (19,0) -- (19.1,1) -- (26.9,1) -- (27,0) -- cycle;
        \fill [gray] (27,0) -- (27.05,1) -- (27.65,1) -- (27.7,0) -- cycle;
        \fill [gray] (28.3,0) -- (28.35,1) -- (28.95,1) -- (29,0) -- cycle;
                \draw (A) -- (B) -- (C) -- (D) -- cycle;
        \draw (5,0) -- (5.1,1) -- (18.9,1) -- (19,0) -- cycle;
        \draw (19,0) -- (19.1,1) -- (26.9,1) -- (27,0) -- cycle;
        \draw (27,0) -- (27.05,1) -- (27.65,1) -- (27.7,0) -- cycle;
        \draw (28.3,0) -- (28.35,1) -- (28.95,1) -- (29,0) -- cycle;
        \draw [dotted, thick] (27.7,0.5) -- (28.3,0.5);
        \coordinate [label = below:{\small $m+1$}] (A) at (2.5,0);
        \coordinate [label = below:{\small $m+2$}] (B) at (12,0);
        \coordinate [label = below:{\small $m+3$}] (C) at (23,0);
        \coordinate [label = below:{\small $m$}] (D) at (27.35,0);
        \coordinate [label = below:{\small $1$}] (G) at (28.65,0);
        \coordinate [label = above:$9B$] (A) at (2.5,1);
        \coordinate [label = above:$26B$] (B) at (12,1);
        \coordinate [label = above:$16B$] (C) at (23,1);
        \coordinate [label = above:$a_m$] (D) at (27.35,1);
        \coordinate [label = above:$a_1$] (G) at (28.65,1);
        \fill [brown] (0,-3.5) -- (0.1,-2.5) -- (18.9,-2.5) -- (19,-3.5) -- cycle;
        \fill [brown] (19,-3.5) -- (19.1,-2.5) -- (21.9,-2.5) -- (22,-3.5) -- cycle;
        \fill [gray] (22,-3.5) -- (22.1,-2.5) -- (22.9,-2.5) -- (23,-3.5) -- cycle;
                \draw (0,-3.5) -- (0.1,-2.5) -- (18.9,-2.5) -- (19,-3.5) -- cycle;
        \draw (19,-3.5) -- (19.1,-2.5) -- (21.9,-2.5) -- (22,-3.5) -- cycle;
        \draw (22,-3.5) -- (22.1,-2.5) -- (22.9,-2.5) -- (23,-3.5) -- cycle;
        \draw (0,-3.5) -- (29,-3.5);
        \coordinate [label = below:{\small $m+4$}] (A) at (9.5,-3.55);
        \coordinate [label = below:{\small $m+5$}] (B) at (20.3,-3.55);
        \coordinate [label = below:{\small $m+6$}] (C) at (22.7,-3.55);
        \coordinate [label = above:$35B$] (G) at (9.5,-2.5);
        \coordinate [label = above:$6B$] (A) at (20.5,-2.5);
        \coordinate [label = above:$2B$] (B) at (22.5,-2.5);
        \draw (0,-4.7) -- (0.1,-4.2) -- (3.9,-4.2) -- (4,-4.7) -- cycle;
  \fill[brown] (0,-4.7) -- (0.1,-4.2) -- (3.9,-4.2) -- (4,-4.7) -- cycle;
  \coordinate [label = below:reclaimer $R_0$] (R0) at (2,-4.7);
  \draw (25,-4.7) -- (25.1,-4.2) -- (28.9,-4.2) -- (29,-4.7) -- cycle;
  \fill[gray] (25,-4.7) -- (25.1,-4.2) -- (28.9,-4.2) -- (29,-4.7) -- cycle;
  \coordinate [label = below:reclaimer $R_1$] (R1) at (27,-4.7);
   \coordinate [label = above:$P_1$] (P1a) at (-.6,0);
   \coordinate [label = above:$P_2$] (P2a) at (-.6,-3.5);
        \end{tikzpicture}
        \begin{tikzpicture}[scale =0.1]
        \centering
        \pgfkeys{
            /pgf/number format/precision=0,
            /pgf/number format/fixed zerofill=true,
            /pgf/number format/fixed
        }
            \node at (35,56) {{makespan: $75B$}};
            \draw[->] (0.0,0.0) -- (76.0,0.0) node[anchor=west] {{\small time}}  ;
            \draw[->] (0.0,0.0) -- (0.0,54.0) node[anchor=south] {{\small space}} ;
            \foreach \y in {0,...,8}
            { \pgfmathsetmacro\result{6*\y}
            \draw ( 0.5, 6*\y) -- ( -0.5, 6*\y) node[left, xshift = -0.05] {{\small{\pgfmathprintnumber{\result}$B$}}};
            }
            \foreach \x in {0,...,7}
            { \pgfmathsetmacro\result{10*\x}
            \draw (\result, .5) -- (\result, -.5) node[below, yshift = -0.02] {{\small{\pgfmathprintnumber{\result}$B$}}};
            }

            \foreach \y in {0,...,26}
                \draw[very thin, loosely dotted] (0.0,2*\y) -- (75,2*\y) ;
            \draw[brown, thick] (0,0) -- (17.5,35) -- (23.5,41) -- (26.5,35) -- (61.5,0) -- (70.5,9) -- (75,0);
            \draw[gray][very thick][dotted](0.0,53) -- (1.5,51);
            \draw[gray][very thick][dotted](73.5,51) -- (75,53);
            \draw[gray, thick, dashed] (1.5,51) -- (17.5,35) -- (21.5,43) -- (23.5,41) -- (26.5,35) -- (52.5,9) -- (73.5,51);
            \draw[gray, thick, dashed] (76,15) -- (80,15);
            \node at (82.5,15) {{$R_1$}};
            \draw[brown, thick] (76,20) -- (80,20);
            \node at (82.5,20) {{$R_0$}};
       \end{tikzpicture}
\caption{Reclaimer scheduling instance and optimal reclaimer
movements used in the proof of Theorem~\ref{thm:oracle_hardness}. The assignment of stockpiles to
reclaimers is indicated by colors, the numbers above the stockpiles are their lengths, and the
numbers below the stockpiles are their indices.}
\label{fig:instance-pf}
\end{figure}
We claim that the smallest possible makespan for this assignment is
$75B$, and that this makespan can be achieved if and only if the
\textsc{Partition} instance is a YES-instance. First, suppose the
instance is a YES-instance, and that $I\subseteq\{1,\ldots,m\}$ is
an index set with $\sum_{i\in I}a_i=B$. In this case, a makespan of
$75B$ is achieved by the following schedule (see
Figure~\ref{fig:instance-pf}).
\begin{itemize}
\item $R_0$ moves from $x=0$ to $x=35B$ without reclaiming anything, and reaches $x=35B$ at time $t=17.5B$. Then it reclaims stockpile $m+5$ from left to right, moves back to $x=35B$ (arriving at time $t=26.5B$), reclaims stockpile $m+4$ from right to left, then stockpile $m+1$ from left to right, and finally returns to its starting point at time $t=75B$.
\item $R_1$ moves from $x=53B$ to $x=51B$ while reclaiming the stockpiles $j$ with $j\in I$, then it reclaims stockpile $m+3$ from right to left, reaching $x=35B$ at time $t=17.5B$. It moves to $x=43B$ without reclaiming anything, reclaims stockpile $m+6$ from right to left, moves to $x=35B$, where it arrives at time $t=26.5B$. Then it reclaims stockpile $m+2$, moves back to $x=51B$, and finally to $x=53B$, reclaiming the stockpiles $j$ with $j\in\{1,\ldots,m\}\setminus I$.
\end{itemize}

Next, we assume the existence of a schedule with a makespan of at
most $75B$. We want to show that this implies that the instance is a
YES-instance. For the sake of contradiction, assume that the
\textsc{Partition} instance does not have a solution. In order to
reclaim the stockpile $m+2$, reclaimer $R_1$ has to spend a time
interval of length $39B$ in the interval $X=[9B,35B]$. It cannot
enter this interval before time $9B$, and the latest possible time
for leaving $X$ is $75B-9B=66B$. This implies that $R_1$ has to
enter $X$ between $t=9B$ and $t=(66-39)B=27B$. Let
$I,I'\subseteq\{1,\ldots,m\}$ be the sets of stockpiles that $R_1$
reclaims before and after its first visit to $X$, respectively. Note
that our assumption on the \textsc{Partition} instance implies
$\sum_{j\in I}a_j\neq B$.
\begin{description}
\item[Case 1.]
$R_1$ enters $X$ at time $t_0<26.5B$. Since $R_0$ cannot finish
reclaiming stockpile $m+5$ and be back at $x=35B$ before time
$26.5B$, this implies that $R_0$ starts reclaiming stockpile $m+5$
after $R_1$ has left $X$. If $R_0$ does not reclaim both stockpiles
$m+1$ and $m+4$ before stockpile $m+5$, then its makespan is at
least $t_0+39B+6B+16B+9B\geqslant 79B$. So we may assume that $R_0$
reclaims stockpiles $m+1$, $m+4$ and $m+5$ in this order. Its
makespan is at least $t_0+39B+6B+20.5B$, hence $t_0\leqslant
(75-65.5)B=9.5B$. This implies that before entering $X$, the maximal
time that $R_1$ can spend in the interval $[51B,53B]$ is
$9.5B-8B=1.5B$. Together with our assumption that $\sum_{j\in
I}a_j\neq B$, this implies $\sum_{j\in I}a_j < B$. On the other
hand, $R_0$ does not leave the interval $[35B,53B]$ before time
$9B+9B/2+41B+3B=57.5B$, and this is the earliest time at which $R_1$
can start reclaiming stockpile $m+3$. Consequently,
$57.5B+16B+B+\frac12\sum_{j\in I'}a_j\leqslant 75B$, i.e.,
$\sum_{j\in I'}a_j<B$, which is a contradiction to the fact that
$\sum_{j \in I \cup I'} a_j = 2B$.

\item[Case 2.]
$R_1$ enters $X$ at time $t_0\in[26.5B,\,27B]$.
$26.5B+39B+16B=81.5B>75B$ implies that stockpile $m+3$ needs to be
reclaimed before $R_1$ enters $X$. From our assumption about the
\textsc{Partition} instance and
\[75B\geqslant 26.5B+39B+8B+B+\frac12\sum_{j\in I'}a_j=74.5B+\frac12\sum_{j\in I'}a_j\]
it follows that $\sum_{j\in I'}a_j<B$, hence $\sum_{j\in I}a_j>B$.
Since $26.5B+39B+6B+20.5B=92B>75B$, we deduce that $R_0$ reclaims
stockpile $m+5$ before $R_1$ enters $X$, i.e., before time $27B$.
This is only possible if $m+5$ is the first stockpile reclaimed by
$R_0$, because $9B+16B+6B>27.5B$. Together with the makespan bound
of $75B$, this implies that $R_0$ enters the interval $[35B,53B]$ at
time $17.5B$ and leaves it at time $26.5B$. From $\sum_{j\in
I}a_j>B$ it follows that $R_1$ cannot finish reclaiming stockpile
$m+3$ at time $17.5B$. This leaves two possibilities.
\begin{description}
\item[Case 2.1.] Stockpile $m+6$ is reclaimed before $R_1$ enters $X$. Then the entering time is at least
\[3B/2+4B+2B+5B+16B=28.5B,\]
which is the required contradiction.
\item[Case 2.2.] Stockpile $m+6$ is reclaimed after $R_1$ enters $X$. Then the makespan of $R_1$ is at least
\[26.5B+39B+3B+2B+5B=75.5B,\]
which is the required contradiction.\qedhere
\end{description}
\end{description}
\end{proof}

\subsection{Precedence constraints}

\subsubsection{Single reclaimer}

In this section, for sake of simplicity, we assume that the stockpiles are indexed by their position in the
precedence chain, i.e., $J = \{1, 2, \ldots, n \}$ and $1
\rightarrow 2 \rightarrow \cdots \rightarrow n$ (where $i
\rightarrow j$ means that stockpile $i$ has to be reclaimed before
stockpile $j$). Furthermore, for convenience, we add two dummy
stockpiles, one at the beginning of the precedence chain (stockpile
0) and one at the end of the precedence chain (stockpile $n+1$) with
$(l_0,r_0) = (l_{n+1},r_{n+1}) = (0,0)$.

\begin{thm}
Determining an optimal schedule for a single reclaimer when the
positions of the stockpiles are given and the stockpiles have to be
reclaimed in a prespecified order can be done in time $O(n)$.
\end{thm}
\begin{proof}
Let $f_r(j)$ be the optimal makespan that can be obtained starting
from the right endpoint of stockpile $j$ and processing stockpiles
$j+1, \ldots, n+1$. Similarly, let $f_l(j)$ to be the optimal
makespan that can be obtained starting from the left endpoint of
stockpile $j$ and processing stockpiles $j+1, \ldots, n+1$.
Naturally, we are interested to $f_r(0) = f_l(0)$. The functions
$f_r$ and $f_l$ can be computed as follows.

\begin{align*}
  f_{r}(j) &=
  \begin{cases}
    0 & \text{if }j=n+1,\\
    \lvert r_{j+1}-l_{j+1}\rvert+\min\left\{\frac{\lvert r_{j}-l_{j+1}\rvert}{s}+f_{r}(j+1),\ \frac{\lvert r_{j}-r_{j+1}\rvert}{s}+f_{l}(j+1)\right\} & \text{if }j\leqslant n.
  \end{cases}\\
  f_{l}(j) &=
  \begin{cases}
     0 & \text{if }j=n+1,\\
    \lvert r_{j+1}-l_{j+1}\rvert+\min\left\{ \frac{\lvert l_{j}-l_{j+1}\rvert}{s}+f_{r}(j+1),\ \frac{\lvert l_{j}-r_{j+1}\rvert}{s}+f_{l}(j+1)\right\} & \text{if }j\leqslant n.
   \end{cases}
\end{align*}
and $f_l(0)$ can be computed backward from $f_l(n)$ and $f_r(n)$ in
$O(n)$ time.
\end{proof}

\subsubsection{Two reclaimers}

Each stockpile must be processed either from left to right or from
right to left by a reclaimer.

\begin{thm}\label{thm:two_pred}
If the travel speed $s$ is an integer, then an optimal schedule for two reclaimers when the positions of the stockpiles are given and the stockpiles have to be reclaimed in a prespecified order can be determined in pseudo-polynomial time, in particular $O(nL^3)$.
\end{thm}
\begin{proof}
We consider $n+1$ stages $j = 0, 1, \ldots, n$ and indicate with
$x_{0}^{(j)}$ and $x_{1}^{(j)}$ the positions of reclaimers $R_{0}$
and $R_{1}$ at stage $j$, respectively. At stage $j = 0$ reclaimers
$R_{0}$ and $R_{1}$ are located at positions $x_{0}^{(0)} = 0$ and
$x_{1}^{(1)} = L$, respectively. At stage $0 < j < n$ one of the
reclaimers has just finished reclaiming stockpile $j$ and the other
reclaimer has repositioned. At stage $n$, the reclaimers $R_{0}$ and
$R_{1}$ move back to positions $0$ and $L$, respectively (taking
$\lvert x_{0}^{(n)}-0\rvert/s$ and $\lvert x_{1}^{(n)}-L\rvert/s$,
respectively).

The system evolves from stage $j$ to stage $j+1$ according to the
following rules. One reclaimer, say $R_{0}$ at position
$x_{0}^{(j)}$, is chosen to move either to point $l_{j+1}$ and
reclaim stockpile $j+1$ ending at point $r_{j+1}$ and taking time
$t=\lvert x_{0}^{(j)}-l_{j+1}\rvert/s + (r_{j+1}-l_{j+1})$ or to
point $r_{j+1}$ and reclaim stockpile $j+1$ ending at point
$l_{j+1}$ and taking time $t = \lvert x_{0}^{(j)}-r_{j+1}\rvert/s +
(r_{j+1}-l_{j+1})$. In either case, the final position of $R_{0}$ is
denoted by $x_{0}^{(j+1)}$. The other reclaimer, in this case
$R_{1}$ at position $x_{1}^{(j)}$, will reposition to a point in the
set $[x_{0}^{(j+1)}, L] \cap [ x_{1}^{(j)} - t s,\,x_{1}^{(j)} + t
s]$. That is, the position $x_{1}^{(j+1)}$ of $R_{1}$ at stage $j+1$
is restricted by the final position of $R_{0}$, by the endpoint of
the pad $L$, and by the maximum distance $ts$ that $R_1$ can travel.

Note that when $x_0^{(j)}$ and $x_1^{(j)}$ are integer points, the
maximum distance $ts$ that $R_1$ can travel is integer, and thus, if
$R_1$ travels the maximum distance, it will end up at an integer
point. Furthermore, if $R_{1}$ does not travel the maximum distance,
it will stop at a stockpile endpoint, and, thus, travel an integer
distance as well (stockpile endpoints are integers) and end up at an
integer point. Because both reclaimers start at integer points, both
reclaimers will be at integer points at every stage.

Let $f(j, x_{0}, x_{1})$ be the minimum makespan that can be
obtained when the two reclaimers start from $x_{0}$ and $x_{1}$,
respectively, to process stockpiles $j+1, \ldots, n$. The optimal
makespan $f(0, 0, L)$ can be computed by backward dynamic
programming.

Let the set of points that can be reached by $R_{0}$ from $x$ in
time $t$ when $R_{1}$ ends in point $y$ be denoted by
\[\Gamma_{0}(x,y,t) = [0, y] \cap [x - ts, x + ts] \cap\mathbb{N}.\]
Similarly, let the set of points that can be reached by $R_{1}$ from
$x$ in time $t$ when $R_{0}$ stops in $y$ be denoted by
\[\Gamma_{1}(x,y,t) = [y,L] \cap [x - ts, x + ts] \cap \mathbb{N}.\]
The recursion is given by $f(n, x_{0}, x_{1})=\max\{x_{0}/s,
(L-x_{1})/s\}$ and $f(j,x_0,x_1)=\min\{F_1,F_2,F_3,F_4\}$ for $j<n$,
where

\begin{align*}
  F_1 &= \frac{|x_{0}-l_{j+1}|}{s}+|r_{j+1}-l_{j+1}| + \min\limits_{x\in\Gamma_{1}(x_{1},r_{j+1},t)}f(j+1,r_{j+1},x),\\
  F_2 &= \frac{|x_{0}-r_{j+1}|}{s}+|r_{j+1}-l_{j+1}| + \min\limits_{x\in\Gamma_{1}(x_{1},l_{j+1},t)}f(j+1,l_{j+1},x),\\
  F_3 &= \frac{|x_{1}-l_{j+1}|}{s}+|r_{j+1}-l_{j+1}| + \min\limits_{x\in\Gamma_{0}(x_{0},r_{j+1},t)}f(j+1,x,r_{j+1}),\\
  F_4 &= \frac{|x_{1}-r_{j+1}|}{s}+|r_{j+1}-l_{j+1}| + \min\limits_{x\in\Gamma_{0}(x_{0},l_{j+1},t)}f(j+1,x,l_{j+1})
\end{align*}
Thus, computing $f(0, 0 , L)$ does not require more than $O(nL^{3})$
time, which is pseudo-polynomial with respect to the instance size.
\end{proof}

\begin{rem}
If the travel speed is rational, say $s=a/b$ for integers
$a\geqslant b$, then multiplying through with $b$ and repeating the
proof of Theorem~\ref{thm:two_pred}, we obtain an optimal solution
in time $O(n(bL)^3)$.
\end{rem}

\section{Reclaimer Scheduling With Positioning Decisions}\label{sec:with}

\subsection{No precedence constraints}
\subsubsection{Single reclaimer}

\begin{thm}\label{thm:pos_single}
Positioning stockpiles and simultaneously determining an optimal
schedule for a single reclaimer is NP-hard.
\end{thm}
\begin{proof}
Reduction from \textsc{Partition} to an instance with pad length
$L=2B$. We map each element $a_i$ to a stockpile $i$ of length $a_i$
and set $s=1$. \textsc{Partition} has a solution if and only if the
stockpiles can be divided over the two pads in such a way that they
take up an equal amount of space, i.e., if the makespan is equal to
$2B$.
\end{proof}

\subsubsection{Two reclaimers}

\begin{thm}\label{thm:pos_two}
Positioning stockpiles and simultaneously determining an optimal
schedule for two reclaimers is NP-hard.
\end{thm}
\begin{proof}
Reduction from \textsc{Partition} to an instance with pad length
$L=2B$. We map each element $a_i$ to a stockpile $i$ of length
$a_i$, add two dummy stockpiles each with length $B$, and set $s=1$.
\textsc{Partition} has a solution if and only if the stockpiles can
be divided equally over the two reclaimers and equally over the two
pads in such a way that they take up an equal amount of space, i.e.,
if the makespan (for both reclaimers) is equal to $2B$.
\end{proof}

\subsection{Precedence constraints}

\subsubsection{Single reclaimer}

\red{Note that our problem is feasible only if it is possible to place all stockpiles on the two pads. An
obvious necessary condition is that the sum of all stockpile lengths is at most $2L$. This is
clearly not sufficient, as $p_1=p_2=p_3=6$ and $L=10$ give an infeasible instance, and if the sum of
the stockpile lengths is equal to $2L$ then it is actually NP-hard to decide if the instance is feasible. With the
stronger assumption that $p_1+\cdots+p_n\leqslant 3L/2$ the problem is always feasible, and in fact we can find an optimal solution as follows.}  
Let $P =p_1+\cdots+p_n$ be the sum of all stockpile lengths, and let
\begin{align*}
  P^t &= \sum_{i=1}^tp_i, & \overline P^t &=\sum_{i=t+1}^np_i= P-P^t
\end{align*}
for $t\in\{0,\ldots,n\}$. 
Let $k$ be the unique index with $P^{k-1}\leqslant \overline P^k$ and $P^k >
\overline P^{k+1}$. The stockpiles are positioned as follows.
\begin{description}
\item[Case 1.] $\min\{ P^k,\, \overline P^{k-1}\} \leqslant L$. If $P^k < \overline P^{k-1}$, then
  place stockpiles $1, \ldots, k$ on pad $P_1$, one after the other starting from 0 and place
  stockpiles $k+1, \ldots, n$ on pad $P_2$, one after the other in reverse order starting from 0. If
  $P^k \geqslant \overline P^{k-1}$, then place stockpiles $1, \ldots, k-1$ on pad $P_1$, one after
  the other starting from 0 and place stockpiles $k, \ldots, n$ on pad $P_2$, one after the other in
  reverse order starting from 0. \red{The stockpiles on pad $P_1$ are reclaimed from left to right and
  then the stockpiles on pad $P_2$ are reclaimed from right to left.}
\item[Case 2.] $\min\{P^k,\,\overline P^{k-1}\} > L$. In this case $p_k>\max\{L/2,\,P^{k-1},\,\overline P^k\}$ and $P-p_{k} < L$ (because we have assumed that $P\leqslant 3L/2$). Place stockpiles $1, \ldots,k-1$ on pad $P_1$, one after the other starting from 0, followed by
$k+1, \ldots, n$ also on pad $P_1$, one after the other in reverse order starting from $P -
p_{k}$. Place stockpile $k$ on pad $P_2$ starting from $\max \{P - 2 p_{k}, 0 \}$. \red{The stockpiles
$1,\ldots,k$ are reclaimed from left to right, and the stockpiles $k+1,\ldots,n$ are reclaimed from right to left.}
\end{description}
The stockpile positions are described more precisely in Algorithm~\ref{alg:fb_positioning}, and the
schedule is completely determined by the positions and the given reclaim directions.
\begin{algorithm}
  \caption{Forward-Backward positioning}\label{alg:fb_positioning}

  \begin{tabbing}
    .....\=.....\=.....\=........................ \kill \\
\textbf{if} $\min \{ P^k,\, \overline P^{k-1} \} \leqslant L$ \textbf{then}\\
\> \textbf{if} $P^k < \overline P^{k-1}$ \textbf{then}\\
\> \> \textbf{for} $i=1,\ldots,k$ \textbf{do} $(l_i,\,r_i)\leftarrow (P^{i-1},\,P^i)$ (on pad 1)\\
\> \> \textbf{for} $i=k+1,\ldots,n$ \textbf{do} $(l_i,\,r_i)\leftarrow (\overline P^i,\,\overline
P^{i-1})$ (on pad 2)\\
\> \textbf{else}\\
\> \> \textbf{for} $i=1,\ldots,k-1$ \textbf{do} $(l_i,\,r_i)\leftarrow (P^{i-1},\,P^i)$ (on pad 1)\\
\> \> \textbf{for} $i=k,\ldots,n$ \textbf{do} $(l_i,\,r_i)\leftarrow (\overline P^i,\,\overline
P^{i-1})$(on pad 2)\\
\textbf{else}\\
\> \textbf{for} $i=1,\ldots,k-1$ \textbf{do} $(l_i,\,r_i)\leftarrow (P^{i-1},\,P^i)$ (on pad 1)\\
\> \textbf{for} $i=k+1,\ldots,n$ \textbf{do} $(l_i,\,r_i)\leftarrow (P^{k-1}+\overline
P^i,\,P^{k-1}+\overline P^{i-1})$ (on pad 1)\\
\> $(l_k,\ r_k)\leftarrow\left(\max\{P - 2 p_k,\, 0 \},\ \max\{ P-p_k,\, p_k \} \right)$ (on pad 2)
  \end{tabbing}
\end{algorithm}
\begin{lem}\label{lem:fb_objective}
For instances with $p_1+\cdots+p_n\leqslant 3L/2$, the stockpile positions determined by
Algorithm~\ref{alg:fb_positioning} with the reclaim directions described above achieve a makespan of
\begin{equation}\label{eq:FB_makespan}
C = P+\min\limits_{0\leqslant t\leqslant n}\left\lvert P^t-\overline P^t\right\rvert/s.
\end{equation}
\end{lem}
\begin{proof}
\red{By definition $k$ is the index for which the minimum in~\eqref{eq:FB_makespan} is obtained. In Case
1, there is exactly one time interval in which the reclaimer moves without reclaiming anything,
namely when it moves from $r_k$ to $r_{k+1}$ or from $r_{k-1}$ to $r_k$ between reclaiming the
stockpiles on pad $P_1$ and the stockpiles on pad $P_2$. This takes time $\left\lvert P^k-\overline
  P^k\right\rvert/s$ and the result follows.}

In Case 2, there are three time intervals in which the reclaimer moves without
reclaiming anything, namely when it moves (1) from $r_{k-1}$ to $l_k$, (2) from $r_k$ to
$r_{k+1}$ and (3) from $l_n$ to $0$. These travel time intervals
have lengths
\begin{align*}
& P^{k-1}-\max\{P-2p_k,\,0\},\\
& \max\{p_k-(P-p_k),\,0\}=\max\{2p_k-P,\,0\},\\
& P^{k-1},
\end{align*}
and this yields a makespan
\[C = P+\frac{P^{k-1}-P+2p_k+P^{k-1}}{s}=P+\frac{2P^k-P}{s}=P+\frac{P^k-\overline P^k}{s}.\qedhere\]
\end{proof}

The next lemma states that the RHS of~(\ref{eq:FB_makespan}) is a lower bound on the makespan.
\begin{lem}\label{lem:fb_lower_bound}
For any placement of the stockpiles and any feasible reclaimer
schedule, the makespan $C$ satisfies
\[ C \geqslant P+\min\limits_{0\leqslant k\leqslant n}\left\lvert P^k-\overline P^k\right\rvert/s.\]
\end{lem}
\begin{proof}
Suppose we have an optimal stockpile placement together with an
optimal reclaimer schedule. Let $x_1$ be the rightmost point reached
by the reclaimer, and let $t_1$ be the time when $x_1$ is reached
for the first time. Let $I$ be the set of stockpiles whose
reclaiming is at or before time $t_1$ and let $J$ be the set of
stockpiles whose reclaiming is finished after time $t_1$. Clearly
$I=\{1,\ldots,k\}$ for some $k$, where $k=0$ corresponds to
$I=\emptyset$. The sets $I$ and $J$ can be partitioned according to
the pads on which the stockpiles are placed: $I=I_1\cup I_2$ where
$I_1$ contains the stockpiles on pad $P_1$, and $I_2$ contains the
stockpiles on pad $P_2$, and similarly for $J=J_1\cup J_2$. Let
\begin{align*}
X_i &= \bigcup\limits_{j\in I_i}[l_j,r_j], & Y_i &= \bigcup\limits_{j\in J_i}[l_j,r_j]
\end{align*}
for $i\in\{1,2\}$. Recall that the total length of a set
$X\subseteq[0,L]$ is denoted by $\ell(X)$. The total stockpile
length equals
\begin{equation}\label{eq:total_length}
  P=\ell(X_1)+\ell(X_2)+\ell(Y_1)+\ell(Y_2).
\end{equation}
Between $t=0$ and $t=t_1$ the reclaimer has to visit (1) the set
$X_1\triangle X_2$ while reclaiming, (2) the set $X_1\cap X_2$ twice
while reclaiming and at least once while moving without reclaiming,
(3) the set $[0,x_1]\setminus(X_1\cup X_2)$ at least once without
reclaiming. This yields
\begin{multline*}t_1\geqslant \ell(X_1)+\ell(X_2)+\frac{\ell(X_1\cap X_2)+x_1-\ell(X_1\cup X_2)}{s}\\
=\ell(X_1)+\ell(X_2)+\frac{x_1-\ell(X_1)-\ell(X_2)+2\ell(X_1\cap X_2)}{s}.
\end{multline*}
Applying the same argument to the time interval $[t_1,C]$ and the
sets $Y_1$ and $Y_2$, we have
\[C-t_1\geqslant \ell(Y_1)+\ell(Y_2)+\frac{x_1-\ell(Y_1)-\ell(Y_2)+2\ell(Y_1\cap Y_2)}{s}.\]
Adding these two inequalities, taking into
account~(\ref{eq:total_length}), we obtain
\begin{equation*}
C\geqslant P+\frac{2(x_1+\ell(X_1\cap X_2)+\ell(Y_1\cap Y_2))-P}{s}.
\end{equation*}
Using $x_1+\ell(X_1\cap X_2)\geqslant\ell(X_1)+\ell(X_2)=P^k$ and $x_1+\ell(Y_1\cap Y_2)\geqslant\ell(Y_1)+\ell(Y_2)=\overline P^k$, we obtain
\[C\geqslant P+\frac{2\max\{P^k,\,\overline P^k\}-P}{s}=P+\left\lvert P^k-\overline P^k\right\rvert/s.\qedhere\]
\end{proof}
\red{
From Algorithm~\ref{alg:fb_positioning}, and Lemmas~\ref{lem:fb_objective} and
~\ref{lem:fb_lower_bound} we get the following theorem.
\begin{thm}\label{thm:fb_optimal}
If the sum of the stockpile lengths is at most $3L/2$ then the problem of finding optimal stockpile
positions and a corresponding schedule for a single reclaimer can be solved in time $O(n)$. 
\end{thm}
}
\subsubsection{Two reclaimers}

\begin{thm}\label{thm:hardness:pos_two_rec_prec}
Positioning stockpiles and simultaneously determining an optimal
schedule for two reclaimers when the stockpiles have to be reclaimed
in a given order is NP-hard.
\end{thm}

To prove the NP-hardness we use 1,6-\textsc{Partition}, the following
variation of \textsc{Partition}:
\begin{description}
\item[1,6-\textsc{Partition}.] Given a set $A =\{a_1, \dots, a_n\}$ of positive integers with $\sum_{i=1}^n a_i = 7B$, can the set $A$ be partitioned into two disjoint subsets $A_1$ and $A_2$ such that $\sum_{a_i \in A_1} a_{i} = B$ and $\sum_{a_i\in A_2} a_{i} = 6B$?
\end{description}
We illustrate the idea of the proof with the following example.

\begin{example}\label{ex:1_6_partition}
Consider the following instance of 1,6-\textsc{Partition}: a set $A
= \{a_1, \dots, a_n\} = \{5, 1, 6, 1, 1, 7\}$ with $\sum_{i=1}^n a_i
= 7B = 21$, i.e., $B = 3$. Create the following instance of the
reclaimer scheduling problem: pad length $L=108$, traveling speed
$s=3$, and a set of $n+4 = 10$ stockpiles of lengths 30, 87, 3, 21,
5, 1, 6, 1, 1, 7, respectively, to be reclaimed in that order. Note
that four special stockpiles have been added that have to be
reclaimed first.

An obvious lower bound on the objective function value is 162, the
sum of the reclaim times of the stockpiles.  Next consider the
stockpile placements and reclaimer assignments shown in Figure~\ref{fig:instance-pf-ex-prec},
i.e., stockpiles 1 and 3 together
with stockpiles 5, 7, and 10 are assigned to $R_0$ and stockpiles 2
and 4 together with stockpiles 6, 8, and 9 are assigned to $R_1$.
Furthermore, let $R_0$ reclaim stockpile 1 going out and
stockpiles 3, 5, 7, and 10 coming back, and let $R_1$
reclaim stockpile 2 going out and stockpiles 4, 6, 8, and 9 coming
back.
\begin{figure}[htb]
\centering
        \begin{tikzpicture}
        \coordinate (A) at (0,1.2);
        \coordinate (B) at (0.05,1.6);
        \coordinate (C) at (4.25,1.6);
        \coordinate (D) at (4.3,1.2);
        \fill [brown] (A) -- (B) -- (C) -- (D) -- cycle;
        \fill [gray] (4.3,1.2) -- (4.35,1.6) -- (6.75,1.6) -- (6.8,1.2) -- cycle;
        \fill [gray] (8.6,1.2) -- (8.65,1.6) -- (8.75,1.6) -- (8.8,1.2) -- cycle;
        \fill [gray] (11,1.2) -- (11.05,1.6) -- (11.15,1.6) -- (11.2,1.2) -- cycle;
        \fill [gray] (11.2,1.2) -- (11.25,1.6) -- (11.35,1.6) -- (11.4,1.2) -- cycle;
        \draw (A) -- (B) -- (C) -- (D) -- cycle;
        \draw (4.3,1.2) -- (4.35,1.6) -- (6.75,1.6) -- (6.8,1.2) -- cycle;
        \draw (8.6,1.2) -- (8.62,1.6) -- (8.78,1.6) -- (8.8,1.2) -- cycle;
        \draw (11,1.2) -- (11.02,1.6) -- (11.18,1.6) -- (11.2,1.2) -- cycle;
        \draw (11.2,1.2) -- (11.22,1.6) -- (11.38,1.6) -- (11.4,1.2) -- cycle;
        \draw[decoration={brace}, decorate] (6.8,1.6) -- (8.6,1.6);
        \coordinate [label = above:{\small $15$}]  (A) at (7.7,1.7);
        \draw[decoration={brace}, decorate] (8.8,1.6) -- (11,1.6);
        \coordinate [label = above:{\small $18$}]  (A) at (9.9,1.7);
        \draw[decoration={brace}, decorate] (11.4,1.6) -- (14,1.6);
        \coordinate [label = above:{\small $21$}]  (A) at (12.7,1.7);
        \draw[decoration={brace}, decorate] (3.1,1.6) -- (4.3,1.6);
        \coordinate [label = above:{\small $9$}]  (A) at (3.7,1.7);
        \draw (0,1.2) -- (14,1.2);
        \coordinate [label = below:{\small $1$}] (A) at (2.15,1.2);
        \coordinate [label = below:{\small $4$}] (B) at (5.55,1.2);
        \coordinate [label = below:{\small $6$}] (C) at (8.7,1.2);
        \coordinate [label = below:{\small $8$}] (D) at (11.08,1.2);
        \coordinate [label = below:{\small $9$}] (E) at (11.32,1.2);
        \coordinate [label = above:{\small $30$}] (A) at (2.15,1.6);
        \coordinate [label = above:{\small $21$}] (B) at (5.55,1.6);
        \coordinate [label = above:{\small $1$}] (C) at (8.7,1.6);
        \coordinate [label = above:{\small $1$}] (D) at (11.1,1.6);
        \coordinate [label = above:{\small $1$}] (E) at (11.3,1.6);
        \draw (0,-.2) -- (14,-.2);
        \fill [brown] (0,-.2) -- (0.05,.2) -- (0.95,.2) -- (1,-.2) --cycle;
        \fill [brown] (1,-.2) -- (1.05,.2) -- (1.85,.2) -- (1.9,-.2) -- cycle;
        \fill [brown] (1.9,-.2) -- (1.95,.2) -- (2.65,.2) -- (2.7,-.2) -- cycle;
        \fill [brown] (2.7,-.2) -- (2.75,.2) -- (3.05,.2) -- (3.1,-.2) -- cycle;
        \fill [gray] (3.1,-.2) -- (3.15,.2) -- (13.95,.2) -- (14,-.2) -- cycle;
        \draw (0,-.2) -- (0.05,.2) -- (0.95,.2) -- (1,-.2) --cycle;
        \draw (1,-.2) -- (1.05,.2) -- (1.85,.2) -- (1.9,-.2) -- cycle;
        \draw (1.9,-.2) -- (1.95,.2) -- (2.65,.2) -- (2.7,-.2) -- cycle;
        \draw (2.7,-.2) -- (2.75,.2) -- (3.05,.2) -- (3.1,-.2) -- cycle;
        \draw (3.1,-.2) -- (3.15,.2) -- (13.95,.2) -- (14,-.2) -- cycle;

        \coordinate [label = below:{\small $3$}] (A) at (2.9,-.2);
        \coordinate [label = below:{\small $5$}] (B) at (2.3,-.2);
        \coordinate [label = below:{\small $7$}] (C) at (1.45,-.2);
        \coordinate [label = below:{\small $10$}] (D) at (.5,-.2);
        \coordinate [label = below:{\small $2$}] (F) at (8.55,-.2);
        \coordinate [label = above:{\small $87$}] (F) at (8.55,.2);
        \coordinate [label = above:{\small $3$}] (A) at (2.9,.2);
        \coordinate [label = above:{\small $5$}] (B) at (2.3,.2);
        \coordinate [label = above:{\small $6$}] (C) at (1.45,.2);
        \coordinate [label = above:{\small $7$}] (D) at (.5,.2);
        \coordinate [label = above:$P_1$] (P1a) at (-.4,1.2);
        \coordinate [label = above:$P_2$] (P2a) at (-.4,0);
        \draw (0,-1.1) -- (0.1,-.8) -- (2.9,-.8) -- (3,-1.1) -- cycle;
        \fill[brown] (0,-1.1) -- (0.1,-.8) -- (2.9,-.8) -- (3,-1.1) -- cycle;
        \coordinate [label = below:reclaimer $R_0$] (R0) at (1.5,-1.1);
        \draw (11,-1.1) -- (11.1,-.8) -- (13.9,-.8) -- (14,-1.1) -- cycle;
        \fill[gray] (11,-1.1) -- (11.1,-.8) -- (13.9,-.8) -- (14,-1.1) -- cycle;
        \coordinate [label = below:reclaimer $R_1$] (R1) at (12.5,-1.1);
   \end{tikzpicture}
\caption{Stockpile positions and reclaimer assignment (indicated by colors) for the reclaimer scheduling instance in
Example~\ref{ex:1_6_partition}. The numbers above the stockpiles are their lengths, and the
numbers below the stockpiles indicate their position in the precedence order.} \label{fig:instance-pf-ex-prec}
\end{figure}
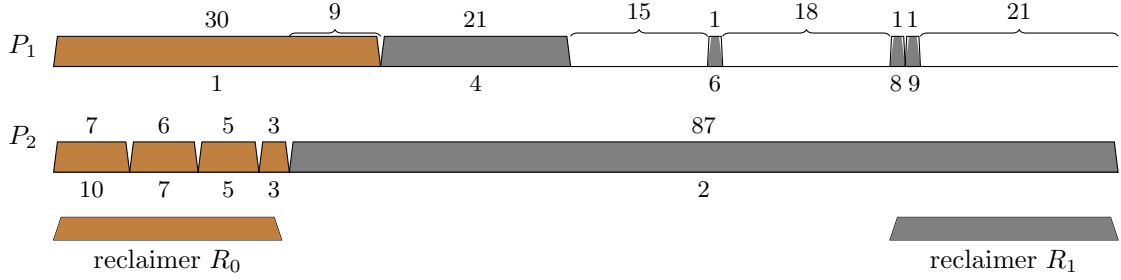

Observe that $R_1$ can complete the reclaiming of
stockpile 4 at time $30 + 87 + 3 + 21 = 141$ (the value 3 corresponds
to the travel time required to go from the left-most point of
stockpile 2 to the right-most point of stockpile 4. Furthermore,
observe that the remaining reclaim time assigned to Reclaimer $R_1$
is 3 and that Reclaimer $R_1$ also has to travel $108 - 30 - 21 - 3 =
54$ to get back to its starting position, which takes 18 for a total
of $3 + 18 = 21$. So Reclaimer $R_1$ can return at time 162 if and
only if it does not incur any waiting time, i.e., it can start
reclaiming stockpiles 6, 8, and 9 as soon as it reaches their
left-most point. Their positions are chosen precisely to make this
happen. For example, the left-most position of stockpile 6 is 66,
i.e., 15 away from 51, which implies that while $R_0$ is
reclaiming stockpile 5, which takes 5, reclaimer $R_1$ can move from
the right-most position of stockpile 4 to the left-most position of
stockpile 6. While $R_1$ reclaims stockpile 6, reclaimer
$R_0$ waits at the right-most position of stockpile 7. And so on.
Finally, observe that the stockpiles 6, 8, and 9 correspond to a
subset $A_1$ with $\sum_{a_i \in A_1} a_{i} = 1 + 1 + 1 = 3 = B$ and
the stockpiles 5, 7, and 10 correspond to a subset $A_2$ with
$\sum_{a_i \in A_2} a_{i} = 5 + 6 + 7 = 18 = 6B$.
\end{example}
More generally, for an instance of 1,6-\textsc{Partition}, we create a
corresponding instance of the reclaimer scheduling problem with pad
length $L = 36B$, traveling speed $s = 3$, and a set of $n+4$
stockpiles of lengths $10B, 29B, B, 7B, a_1, \ldots, a_n$,
respectively, to be reclaimed in that order. We will show that the
instance is a yes-instance of 1,6-\textsc{Partition} if and only if there
exists a reclaimer schedule in which both reclaimers return to their
starting positions at time $54B$.

\begin{proof} Suppose the instance of 1,6-\textsc{Partition} is a
yes-instance, then the placements and assignments shown in Figure
\ref{fig:instance-pf-prec}, i.e., stockpiles 1 and 3 together with
the stockpiles corresponding to the subset $A_1$ are assigned to
$R_0$ and stockpiles 2 and 4 together with the stockpiles
corresponding to the subset $A_2$ are assigned to $R_1$, $(l_1, r_1)
=(0,10B)$, $(l_2, r_2) = (7B,36B)$, $(l_3, r_3) = (6B, 7B)$, and
$(l_4, r_4) = (10B, 17B)$, the stockpiles in $A_2$ are placed on pad $P_2$
 in the interval $[0,6B]$ and the stockpiles in $A_1$ are placed on
pad $P_1$ in the interval $[17B, 36B]$ in such a way that the
distance between two consecutive stockpiles $i$ and $j$ is $3
\sum_{k=i+1}^{j-1}a_{k}$. It can easily be verified that schedule
$S^*$ has an objective function value equal to the lower bound of
$54B$.

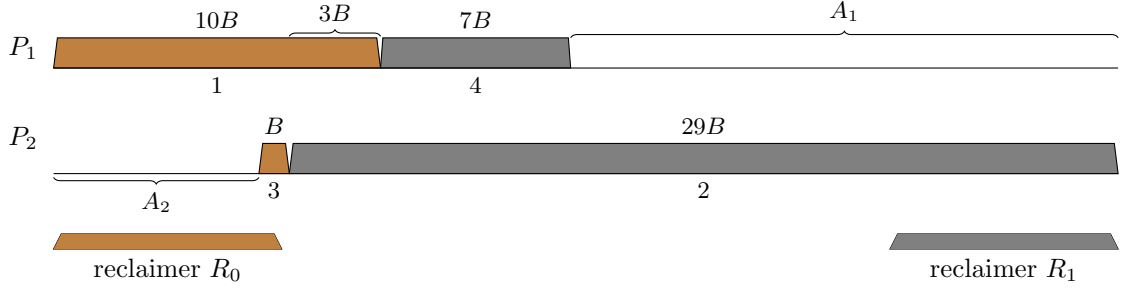
\begin{figure}[htb]
\centering
         \begin{tikzpicture}
        \coordinate (A) at (0,1.2);
        \coordinate (B) at (0.05,1.6);
        \coordinate (C) at (4.25,1.6);
        \coordinate (D) at (4.3,1.2);
        \fill [brown] (A) -- (B) -- (C) -- (D) -- cycle;
        \fill [gray] (4.3,1.2) -- (4.33,1.6) -- (6.77,1.6) -- (6.8,1.2) -- cycle;
        \draw (A) -- (B) -- (C) -- (D) -- cycle;
        \draw (4.3,1.2) -- (4.33,1.6) -- (6.77,1.6) -- (6.8,1.2) -- cycle;
        \draw[decoration={brace}, decorate] (6.8,1.6) -- (14,1.6);
        \coordinate [label = above:{\small $A_1$}]  (A) at (10.4,1.7);
        \draw[decoration={brace}, decorate] (3.1,1.62) -- (4.3,1.62);
        \coordinate [label = above:{\small $3B$}]  (A) at (3.7,1.7);
        \draw (0,1.2) -- (14,1.2);
        \coordinate [label = below:{\small $1$}] (A) at (2.15,1.2);
        \coordinate [label = below:{\small $4$}] (B) at (5.55,1.2);
        \coordinate [label = above:{\small $10B$}] (A) at (2.15,1.6);
        \coordinate [label = above:{\small $7B$}] (B) at (5.55,1.6);
        \draw (0,-.2) -- (14,-.2);
        \fill [brown] (2.7,-.2) -- (2.75,.2) -- (3.05,.2) -- (3.1,-.2) -- cycle;
        \fill [gray] (3.1,-.2) -- (3.15,.2) -- (13.95,.2) -- (14,-.2) -- cycle;
        \draw (2.7,-.2) -- (2.75,.2) -- (3.05,.2) -- (3.1,-.2) -- cycle;
        \draw (3.1,-.2) -- (3.15,.2) -- (13.95,.2) -- (14,-.2) -- cycle;

        \coordinate [label = below:{\small $3$}] (A) at (2.9,-.2);
        \coordinate [label = below:{\small $2$}] (F) at (8.55,-.2);
        \coordinate [label = above:{\small $29B$}] (F) at (8.55,.2);
        \coordinate [label = above:{\small $B$}] (A) at (2.9,.2);
        \draw[decoration={brace}, decorate] (2.7,-.25) -- (0,-.25);
        \coordinate [label = below:{\small $A_2$}]  (A) at (1.35,-.33);
        \coordinate [label = above:$P_1$] (P1a) at (-.4,1.2);
        \coordinate [label = above:$P_2$] (P2a) at (-.4,0);
        \draw (0,-1.2) -- (0.1,-1) -- (2.9,-1) -- (3,-1.2) -- cycle;
        \fill[brown] (0,-1.2) -- (0.1,-1) -- (2.9,-1) -- (3,-1.2) -- cycle;
        \coordinate [label = below:reclaimer $R_0$] (R0) at (1.5,-1.2);
        \draw (11,-1.2) -- (11.1,-1) -- (13.9,-1) -- (14,-1.2) -- cycle;
        \fill[gray] (11,-1.2) -- (11.1,-1) -- (13.9,-1) -- (14,-1.2) -- cycle;
        \coordinate [label = below:reclaimer $R_1$] (R1) at (12.5,-1.2);
        \end{tikzpicture}
 \caption{Stockpile positions and reclaimer assignment (indicated by colors) for the reclaimer
   scheduling instance in the proof of Theorem~\ref{thm:hardness:pos_two_rec_prec}.} \label{fig:instance-pf-prec}
\end{figure}

Next, we prove that lower bound of $54B$ is not achievable if (1)
stockpiles 1, 2, 3, and 4 are placed differently or (2) the instance
of 1,6-\textsc{Partition} is a no-instance.

First, observe that to achieve the lower bound there cannot be any
time between the end of the reclaiming of one stockpile and the
start of the reclaiming of the next stockpile. Furthermore, w.l.o.g., we
can assume that stockpile 1 is assigned to $R_0$ with placement
$(l_1, r_1) = (0, 10B)$.

\vspace{12pt} \noindent \textbf{Claim 1: The lower bound of $54B$
cannot be achieved if stockpiles 1 and 2 are assigned to the same
reclaimer.}

\noindent This is obvious because the stockpiles cannot fit together
on a single path and have different lengths. Therefore, stockpile 2
has to be assigned to $R_1$ and placed on pad $P_2$. It also follows
that stockpile 2 has to be reclaimed from right to left.

\vspace{12pt} \noindent \textbf{Claim 2: The lower bound of $54B$
cannot be achieved if stockpile 3 is assigned to $R_1$ or stockpile
4 to $R_0$.}

To avoid time between the end of reclaiming of stockpile 2 and the
start of reclaiming of stockpile 3, stockpile 3 has to be placed on
pad $P_2$ to the left of stockpile 2. As a consequence, stockpile 4
has to be placed on pad $P_1$, because there is not enough space
left to place it on pad $P_2$. As a result, stockpile 3 and
stockpile 4 cannot be reclaimed by the same reclaimer, because the
right-most position of stockpile 3 will be less than or equal to
$7B$ and the left-most position of stockpile 4 will be greater than
or equal to $10B$. Furthermore, if stockpile 4 would be assigned to
Reclaimer $R_0$, then, because Reclaimer $R_0$ always has to be to
the left of Reclaimer $R_1$, it is unavoidable to incur travel time
before the start of the reclaiming of stockpile 4. Thus, stockpile 4
has to be assigned to $R_1$ and stockpile 3 to $R_0$.

\vspace{12pt} \noindent \textbf{Claim 3: The lower bound of $54B$
cannot be achieved unless the travel time between $l_2$ and $l_4$ is
exactly $B$.}

Since $l_2 \leqslant 7B$ and $l_4 \geqslant 10B$, we have $\frac{l_4 - l_2}{3}
\geqslant B$. Since stockpile 3 has length $B$ and has to be reclaimed
between stockpile 2 and 4, if $l_4 > 10B$ or $l_2 < 7B$, then there
will be travel time of at least $\frac{l_4 - l_2 - B}{3}$ in the
schedule.

From the above three claims it follows that to be able to achieve
the lower bound of $54B$, stockpiles 1 and 3 have to be assigned to
$R_0$, stockpiles 2 and 4 have to be assigned to $R_1$, and the four
stockpiles have to be placed as follows: $(l_1, r_1) = (0, 10B), \
(l_2, r_2) = (7B, 36B), (l_3, r_3) = (6B, 7B)$, and $(l_4, r_4) =
(10B, 17B)$.

\vspace{12pt} \noindent \textbf{Claim 4: The lower bound of $54B$ is
achievable iff the instance of 1,6-\textsc{Partition} is a yes-instance.}

Observe that $R_1$ can complete the reclaiming of stockpile 4 at
time $47B$ and at that time will be at position $17B$. The remaining
space of $19B$ on pad $P_1$ has to be allocated to stockpiles, say
$x$, and unoccupied space, say $y$. To reach its starting position
at or before time $54B$, the time spend on reclaiming, i.e., $x$,
and the time spend on traveling, i.e., $y/3$ should be less than or
equal to $7B$. Thus, we have
\begin{align*}
x + y & = 19B \\
3x + y & \leqslant 21B,
\end{align*}
which implies $x \leqslant B$, i.e., the stockpiles placed on pad
$P_1$ should take up no more space than $B$. However, given that the
remaining space available for the placement of stockpiles on pad
$P_2$ is $6B$, this implies that the stockpiles corresponding to
$a_1, a_2, \ldots, a_n$ with $\sum_{j=1}^n a_j = 7B$ have to be
partitioned into two subsets $A_1$ and $A_2$ with $\sum_{a_i \in
A_1} a_{i} = B$ and $\sum_{a_i \in A_2} a_{i} = 6B$, i.e., the
instance of 1,6-\textsc{Partition} is a yes-instance. 
\end{proof}

\section{Final Remarks}
\label{sec:final}

We have studied a number of variants of an abstract scheduling
problem inspired by the scheduling of reclaimers in the stockyard of
a coal export terminal. 
\red{We leave the following open question for future work.
\begin{enumerate}
\item For the following problems, we used reduction from \textsc{Partition} to prove that they are NP-hard, but we did not decide if they
  are strongly NP-hard.
  \begin{itemize}
  \item Find an optimal schedule for two reclaimers with given stockpile positions and arbitrary
    reclaim order (Theorem~\ref{thm:np_hard_1}).
  \item Find an optimal schedule for two reclaimers with given stockpile positions and and given
    assignment of stockpiles to reclaimers (Theorem~\ref{thm:oracle_hardness}).
  \item Find optimal stockpile positions and reclaimer schedules for two reclaimers with
    arbitrary reclaim order (Theorem~\ref{thm:pos_two}).
  \item Find optimal stockpile positions and reclaimer schedules for two reclaimers with given reclaim
    order (Theorem~\ref{thm:hardness:pos_two_rec_prec}). 
  \end{itemize}
We conjecture that these problems are not strongly NP-hard.
\item We described a pseudo-polynomial algorithm for the problem to schedule two reclaimers for
  given stockpile positions and given reclaim order (Theorem~\ref{thm:two_pred}). We conjecture that
  this problem can actually be solved in polynomial time.   
\end{enumerate}}
One important aspect of the real-life
reclaimer scheduling problem, which is ignored so far, is its
dynamic nature. Vessels arrive over time, and, as a result, the
stockpiles that need be stacked and reclaimed are not all known at
the start of the planning horizon (and do not all fit together on
the pads). We are currently investigating multi-vessel variants of
the problems studied in this paper that explicitly take into account
the time dimension of the reclaimer scheduling problem.


\end{document}